\documentclass[11pt,english]{article}
\usepackage{amsmath,amssymb,amsthm}
\usepackage{multicol}
\usepackage[margin=1in]{geometry}
\usepackage{graphicx,color}
\usepackage{enumitem}
\usepackage{fullpage}
\usepackage[noblocks]{authblk}
\usepackage{tcolorbox}
\usepackage{babel}
\usepackage{wrapfig}
\usepackage{MnSymbol}
\usepackage{mdframed}
\usepackage{sidecap}
\usepackage{algorithm}
\usepackage{algpseudocode}
\usepackage{multirow}
\usepackage{tablefootnote}
\usepackage{thm-restate}
\usepackage{bbding}
\usepackage{pifont}
\usepackage{nicefrac}
\usepackage{mathrsfs}
\usepackage{tikz}
\catcode`\@ = 11
\newdimen\@InsertBoxMargin
\newcount\@numlines    
\newcount\@linesleft   
\def\ParShape{%
    \@numlines = 0
    \def\@parshapedata{ }
    \afterassignment\@beginParShape
    \@linesleft
}%
\def\@beginParShape{%
    \ifnum \@linesleft = 0
      \let\@whatnext = \@endParShape
    \else
      \let\@whatnext = \@readnextline
    \fi
    \@whatnext
}%
\def\@endParShape{%
    \global\parshape = \@numlines \@parshapedata
}%
\def\@readnextline#1 #2 #3 {
    \ifnum #1 > 0
      \bgroup  
        \dimen0 = \hsize
        \advance \dimen0 by -#2  
        \advance \dimen0 by -#3  
        \count0 = 0
        \loop
          \global\edef\@parshapedata{%
            \@parshapedata    
            #2                
            \space            
            \the\dimen0       
            \space            
          }%
          \advance \count0 by 1
          \ifnum \count0 < #1
        \repeat
      \egroup
      \advance \@numlines by #1
    \fi
    \advance \@linesleft by -1
    \@beginParShape
}%
\newbox\@boxcontent     
\newcount\@numnormal    
\newdimen\@framewidth   
\newdimen\@wherebottom  
\newif\if@byframe       
\@byframefalse
\def\InsertBoxC#1{%
  \leavevmode
  \vadjust{
    \vskip \@InsertBoxMargin
    \hbox to \hsize{\hss#1\hss}
    \vskip \@InsertBoxMargin
  }%
}%
\def\InsertBoxL#1#2{%
  \@numnormal = #1
  \setbox\@boxcontent = \hbox{#2}%
  \let\@side = 0
  \futurelet \@optionalparameter \@InsertBox
}
\def\InsertBoxR#1#2{%
  \@numnormal = #1
  \setbox\@boxcontent = \hbox{#2}%
  \let\@side = 1
  \futurelet \@optionalparameter \@InsertBox
}%
\def\@InsertBox{%
  \ifx \@optionalparameter [
    \let\@whatnext = \@@InsertBoxCorrection
  \else
    \let\@whatnext = \@@InsertBoxNoCorrection
  \fi
  \@whatnext
}%
\def\@@InsertBoxCorrection[#1]{%
  \ifx \@side 0
    \@@InsertBox{#1}{0}{{\the\@framewidth} 0cm}%
  \else
    \@@InsertBox{#1}{1}{0cm {\the\@framewidth}}%
  \fi
}%
\def\@@InsertBoxNoCorrection{%
  \@@InsertBoxCorrection[0]%
}%
\def\@@InsertBox#1#2#3{%
  \MoveBelowBox
  \@byframetrue
  \@wherebottom = \baselineskip
  \multiply \@wherebottom by \@numnormal
  \advance \@wherebottom by 2\@InsertBoxMargin
  \advance \@wherebottom by \ht\@boxcontent
  \advance \@wherebottom by \pagetotal
  \ifdim \pagetotal = 0cm
    \advance \@wherebottom by -\baselineskip  
  \fi
  \advance \@wherebottom by #1\baselineskip
  \@framewidth = \wd\@boxcontent
  \advance \@framewidth by \@InsertBoxMargin
  \bgroup  
    \ifdim \pagetotal = 0cm
      \dimen0 = \vsize
    \else
      \dimen0 = \pagegoal
    \fi
    \ifdim \@wherebottom > \dimen0
      \immediate\write16{+--------------------------------------------------------------+}%
      \immediate\write16{| The box will not fit in the page. Please, re-edit your text. |}%
      \immediate\write16{+--------------------------------------------------------------+}%
      \vrule width \overfullrule
    \fi
  \egroup
  \prevgraf = 0
  \vbox to 0cm{%
    \dimen0 = \baselineskip
    \multiply \dimen0 by \@numnormal
    \advance \dimen0 by -\baselineskip
    \setbox0 = \hbox{y}%
    \vskip \dp0
    \vskip \dimen0
    \vskip \@InsertBoxMargin
    \ifnum #2 = 1
      \vtop{\noindent \hbox to \hsize{\hss \box\@boxcontent}}%
    \else
      \vtop{\noindent \box\@boxcontent}%
    \fi
    \vss
  }%
  \vglue -\parskip
  \vskip -\baselineskip
  \everypar = {%
    \ifdim \pagetotal < \@wherebottom
      \bgroup  
        \dimen0 = \@wherebottom
        \advance \dimen0 by -\pagetotal
        \divide \dimen0 by \baselineskip
        \count1 = \dimen0
        \advance \count1 by 1
        \advance \count1 by -\@numnormal
        \ifnum #2 = 1
          \ParShape = 3
                      {\the\@numnormal}   0cm   0cm
                      {\the\count1}       0cm   {\the\@framewidth}
                      1                   0cm   0cm
        \else
          \ParShape = 3
                      {\the\@numnormal}   0cm                  0cm
                      {\the\count1}       {\the\@framewidth}   0cm
                      1                   0cm                  0cm
        \fi
      \egroup
    \else
      \@restore@    
    \fi
  }%
  \def\par{%
      \endgraf
      \global\advance \@numnormal by -\prevgraf
      \ifnum \@numnormal < 0
        \global\@numnormal = 0
      \fi
      \prevgraf = 0
  }%
}%
\def\MoveBelowBox{%
  \par
  \if@byframe
    \global\advance \@wherebottom by -\pagetotal
    \ifdim \@wherebottom > 0cm
      \vskip \@wherebottom
    \fi
    \@restore@
  \fi
}%
\def\@restore@{%
    \global\@wherebottom = 0cm
    \global\@byframefalse
    \global\everypar = {}%
    \global\let \par = \endgraf
    \global\parshape = 1 0cm \hsize
}%
\ifx \documentclass \@Dont@Know@What@It@Is@
\else
  \let \pageno = \c@page
\fi

\catcode`\@ = 12

\usetikzlibrary{angles, arrows,
	calc,
	quotes,
}

\DeclareMathAlphabet{\mathdutchcal}{U}{dutchcal}{m}{n}

\definecolor{Darkblue}{rgb}{0,0,0.4}
\definecolor{Brown}{cmyk}{0,0.61,1.,0.60}
\definecolor{Purple}{cmyk}{0.45,0.86,0,0}
\definecolor{Darkgreen}{rgb}{0.133,0.543,0.133}

\usepackage[colorlinks,linkcolor=Darkblue,filecolor=blue,citecolor=blue,urlcolor=Darkblue,pagebackref]{hyperref}
\usepackage[nameinlink]{cleveref}

\usepackage[colorinlistoftodos,prependcaption,textsize=tiny]{todonotes}
\newif\ifdraft 
\draftfalse

\newtheorem{theorem}{Theorem}
\newtheorem{lemma}{Lemma}

\newtheorem{definition}{Definition}
\newtheorem{claim}{Claim}
\newtheorem{observation}{Observation}

\newtheorem{remark}{Remark}

\let\int\undefined

\newcommand{\polylog}{\mathrm{polylog}}
\newcommand{\poly}{\mathrm{poly}}

\newcommand{\int}{\mathsf{Int}}

\newcommand{\lca}{\mathcal{LCA}}

\newcommand{\diam}{\mathsf{diam}}

\newcommand{\appr}[2]{\left[#2\right]_{#1}}
\newcommand{\inter}{\mathsf{Int}}
\newcommand{\exter}{\mathsf{Ext}}
\newcommand{\defi}{\stackrel{\text{\tiny{def.}}}{=}}
\newcommand{\smallfont}[1]{\scriptscriptstyle{#1}}

\newcommand{\norm}[1]{\Vert #1 \Vert}
\renewcommand{\vec}[1]{\mathbf{#1}}

 \newcommand{\eps}{\epsilon}

\DeclareMathOperator*{\argmin}{arg\,min}

\def\eps{\epsilon}

\newcommand{\initOneLiners}{%
	\setlength{\itemsep}{0pt}
	\setlength{\parsep }{0pt}
	\setlength{\topsep }{0pt}
}

\title{Approximate Distance Oracles for Planar Graphs with Subpolynomial Error Dependency}

\author{Hung Le}
\affil{University of Massachusetts at Amherst}

\date{}
\begin{document}
\maketitle
\begin{abstract}
Thorup [FOCS'01, JACM'04] and Klein [SODA'01] independently showed that there exists a $(1+\eps)$-approximate distance oracle for planar graphs with $O(n (\log n)\eps^{-1})$ space and $O(\eps^{-1})$ query time. While the dependency on $n$ is nearly linear, the space-query product of their oracles depend \emph{quadratically} on $1/\eps$. Many follow-up results either improved the space \emph{or} the query time of the oracles while having the same, sometimes worst, dependency on $1/\eps$. Kawarabayashi, Sommer, and Thorup [SODA'13]  were the first to improve the dependency on $1/\eps$ from quadratic to \emph{nearly linear} (at the cost of $\log^*(n)$ factors). It is plausible to conjecture that the linear dependency on $1/\eps$ is optimal: for many known distance-related problems in planar graphs,  it was proved that the dependency on $1/\eps$ is at least linear. 

In this work, we disprove this conjecture by reducing the dependency of the space-query product on $1/\eps$ from linear all the way down to \emph{subpolynomial} $(1/\eps)^{o(1)}$. More precisely, we construct an oracle with $O(n\log(n)(\eps^{-o(1)} + \log^*n))$  space and $\log^{2+o(1)}(1/\eps)$ query time. Our construction is the culmination of several different ideas developed over the past two decades.

\end{abstract}
\pagebreak
{\small \setcounter{tocdepth}{2} \tableofcontents}
\newpage
\pagenumbering{arabic}

\section{Introduction}

Computing distances is one of the most fundamental primitives in graph algorithms. Approximate distance oracle is a data structure invented specifically for this purpose. A \emph{$t$-approximate distance oracle} of an edge-weighted and undirected\footnote{All oracles in this paper are for undirected and edge-weighted graphs unless mentioned otherwise.} graph $G = (V,E,w)$ is a data structure that given any two vertices $u$ and $v$, return an approximate distance $\hat{d}(u,v)$ such that $d_G(u,v)\leq \hat{d}(u,v)\leq t\cdot d_G(u,v)$. The breakthrough result of Thorup and Zwick~\cite{TZ05} gives a $(2k-1)$-approximate distance oracle for undirected $n$-vertex graphs with $O(kn^{1+1/k})$ space and $O(k)$ query time for any $k\geq 1$. However, reducing the distance error to smaller than a factor of $3$ requires $\Omega(n^2)$~\cite{TZ05} space for dense graphs. In many practical applications, it is desirable to have the distance error as close to $1$ as possible. Constructing an approximate distance oracle with such an error guarantee requires exploiting specific structures of input graphs. Planarity is a natural structure that has been extensively studied for decades; it is used to model road networks on which querying distances is a central problem. 

Thorup~\cite{Thorup04} and Klein~\cite{Klein02} independently constructed a $(1+\eps)$-approximate distance oracle with $O(n\log(n)\eps^{-1})$ space and $O(\eps^{-1})$ query time for any $\eps \in (0,1)$; the construction time of their oracle is $O(n\log^3(n)/\eps^2)$ (see Theorem 3.19 in~\cite{Thorup04}). Their results triggered many follow-up papers over the past two decades; we can generally divide them into two directions. One direction assumes that $\eps$ is a fixed constant and aims to improve the dependency on $n$~\cite{WulffNilsen16,LW21}, which culminated in a distance oracle with $O(n)$ space and $O(1)$ query time by  Le and Wulff-Nilsen~\cite{LW21}. However, the dependency on $1/\eps$ of their oracle's  space-query product is $O(\eps^{-4})$. 

Another direction  focuses on reducing the quadratic dependency on $1/\eps$ of the space-query product. In practical applications such as logistics and planning, a reduction of $1\%$ distance error could lead to a huge economic saving. In such scenarios, $\eps$ is very small and potentially depends on $n$. In an extreme regime, such as $\eps \in [1/\sqrt{n}, 1/n]$, the quadratic dependency on $1/\eps$ implies a quadratic dependency on $n$, making the oracle trivial. Even in a moderately small regime, for example, $\eps = 1/\poly(\log n)$, the dependency on $1/\eps$ remains a dominating factor. Therefore, it is of both theoretical and practical interest to reduce the quadratic dependency on $1/\eps$.  

Kawarabayashi, Klein, and Sommer~\cite{KKS11} argued that ``a very low space requirement is essential'',  and gave a $(1+\eps)$-approximate oracle with truly linear space $O(n)$ and query time $O(\eps^{-2}\log(n))$. While the space is information-theoretically optimal, the query time is blown up by a factor $O(\eps^{-1}\log n)$, making the space-query product worst than that of Thorup~\cite{Thorup04} and Klein~\cite{Klein02}.   The first real improvement was achieved by Kawarabayashi, Sommer, and Thorup~\cite{KST13} who constructed a $(1+\eps)$-approximate oracle with $O(n\log n\log\log(1/\eps)\log^* n)$ space and $O(\eps^{-1}\log^2(1/\eps)\log\log(1/\eps) \log^*n)$ query time. Ignoring $\log^*(n)$ and $\polylog(1/\eps)$ factors, the space bound is $O(n\log n)$ while the query time is $O(1/\eps)$, giving a quadratic improvement in $1/\eps$ dependency in the space-query product.   And yet a decade has passed, and the improvement of  Kawarabayashi, Sommer, and Thorup~\cite{KST13} remains state-of-the-art. 

Recent works instead focus on improving the dependency on $1/\eps$ of the \emph{query time} at the cost of a larger space bound. Gu and Xu~\cite{GX19} constructed an oracle with $O(1)$ query time and $O(n\log(n)(\log(n/\eps) + 2^{O(1/\eps)}))$ space; their space bound is exponential in $1/\eps$. Chan and Skrepetos~\cite{CS19}, using the Voronoi diam technique of Cabello~\cite{Cabello18}, improved the space of the oracle of Gu and Xu~\cite{GX19} to a (large) polynomial in $1/\eps$ at the cost of a slightly worst query time $O(\log(1/\eps))$. However, the construction of Chan and Skrepetos~\cite{CS19} is randomized. Li and Parter~\cite{LP19} devised the VC-dimension technique to  reduce the space of the oracle by  Gu and Xu~\cite{GX19}; the space of their oracle is not explicitly computed but remains polynomial in $1/\eps$, and the query time, while is not explicitly mentioned, is $O(1)$. See \Cref{tab:approx-oracle} for a summary. 

\renewcommand{\arraystretch}{1.3}
\begin{table}[]
	\vspace{-30pt}
	\begin{center}
		\begin{tabular}{|l|l|l|}
			\hline
			Space                                       & Query time             & Reference                                                  \\ \hline
			$O(n\log(n)/\eps)$ & $O(1/\eps)$ & \cite{Thorup04,Klein02}\\ \hline
			$O(n)$ & $O(\log(n)/\eps^2)$ & \cite{KKS11}\\ \hline
			$O(n\log n\log\log(1/\eps)\log^* n)$ & $O((1/\eps)\log^2(1/\eps)\log\log(1/\eps) \log^*n)$ & \cite{KST13}\\ \hline
			$O(n\log(n)(\log(n/\eps) + 2^{O(1/\eps)}))$ & $O(1)$ & \cite{GX19}\\\hline
		 \shortstack{$O(n\log^2(n)\eps^{-1} +n\log(n)\eps^{-(4+\delta)})$\\  any fixed $\delta > 0$}	  & $O(\log(1/\eps))$ & \cite{CS19}\\\hline
			\shortstack{$O(n\log^2(n)(1/\eps)^c)^{++}$ \\ for some constant $c\geq 6$} & $O(1)^{++}$ & \cite{LP19}\\\hline
			$O(n\log(n)((1/\eps)^{o(1)} + \log^*n))$ & $\log^{2+o(1)}(1/\eps)$ & \Cref{thm:main}\\\hline
			$O(n\log(n)(\log^{2+o(1)}(1/\eps) + \log^*n))$ & $(1/\eps)^{o(1)}$ & \Cref{thm:main}\\\hline
		\end{tabular}
	\end{center}
	\caption{Space and query time of known $(1+\epsilon)$-approximate distance oracles for \emph{undirected} planar graphs. The bounds marked $^{++}$ are our own estimation following the description in \cite{LP19}; these bounds are not explicitly computed in \cite{LP19}. The last two rows are our results.}
	\label{tab:approx-oracle}
\end{table}
\renewcommand{\arraystretch}{1}

On the other hand, it is plausible to conjecture that the linear dependency on $1/\eps$ of the space-query product is optimal. For some distance-based problems in planar graphs where one seeks to have structures that preserve distances approximately with an error parameter $\eps$, such as light $(1+\eps)$-spanners~\cite{ADDJS93} or treewidth embedding with low (additive) distortion~\cite{FKS19,CFKL20,FL21B} or approximate planar emulators~\cite{CKT22}, it was proved that the dependency on $1/\eps$ of the output's quality (such as lightness or treewidth or the number of edges) is $\Omega(1/\eps)$. Another related problem is $(1+\eps)$-approximate distance labeling scheme where the best-known scheme~\cite{Thorup04} has labels of size $O(\log(n)/\eps)$; again, the dependency on $1/\eps$ is linear. For all problems seeking some kinds of $(1+\eps)$ approximation in planar graphs that we are aware of, none of them has a sublinear the dependency on $1/\eps$ (though there exist such problems for trees~\cite{GKKPP01,FGNW17}, which form a very restricted subclass of planar graphs).  Furthermore, the constructions of all known approximate oracles rely on the same fundamental building block using shortest path separators pioneered by Thorup~\cite{Thorup04} and Klein~\cite{Klein02}. More precisely, for each shortest path in the separator, one marks  $1/\eps$ vertices on the  shortest path to serve as \emph{portals} for computing approximate distances. This $1/\eps$ factor creeps into the space and/or the query time,  which makes the linear dependency seem unavoidable.

In this work, we break the long-standing linear dependency on $1/\eps$ in space-query product of approximate distance oracles for the first time. Indeed, we improve the dependency of the space-query product on $1/\eps$ from linear all the way down to \emph{subpolynomial} $(1/\eps)^{o(1)}$.

\begin{restatable}{theorem}{Main}\label{thm:main}Let $\eps \in (0,1)$ be positive parameter and $G = (V,E,w)$ be an undirected, edge-weighted planar graphs with $n$ vertices. We can construct in $O(n\poly(\log(n),1/\eps))$ time a $(1+\eps)$-approximate distance oracle that has:
\begin{itemize}
		\item[(1)]  $O(n\log(n)((1/\eps)^{o(1)} + \log^*n))$ space  and  $\log^{2+o(1)}(1/\eps)$ query time or 
	\item[(2)] $O(n\log(n)(\log^{2+o(1)}(1/\eps) + \log^*n))$ space  and  $(1/\eps)^{o(1)}$ query time.	
\end{itemize}
\end{restatable}

Given the aforementioned lower bounds for related problems,  we find the result in \Cref{thm:main} rather surprising. It opens the real possibility that the dependency on $1/\eps$ could be sublinear or even subpolynomial for distance-related problems where there is no currently-known linear lower bound, such as computing $(1+\eps)$-approximate diameter~\cite{WY16,CS19} or $(1+\eps)$-approximate distance labelings~\cite{Thorup04}. For approximate distance labelings, by a reduction from exact distance labeling~\cite{GPPR04,Thorup04}, the lower bound dependency on $1/\eps$ one can show is $\Omega((1/\eps)^{1/3})$. Furthermore, our technique described below might also be used to progress on these problems.
 
 We remark that we do not attempt to minimize the $\poly(1/\eps,\log(n))$ factor in the construction time in \Cref{thm:main}.   In the following section, we review previous techniques and give an overview of our construction.

\subsection{Previous and Our Techniques}

A conceptual contribution of our work is to view approximate distance oracle constructions through the lens of \emph{local portalization} vs \emph{global portalization}. This view explains  why constructions developed over the past two decades fail to break the linear, sometimes quadratic, dependency on $1/\eps$ in the space-query product. Through this view, we  identify key strengths and weaknesses of each construction, and the challenges in overcoming the linear $1/\eps$ barrier. We then design a framework that could exploit the strengths of all of them through which we obtain a sublinear bound $(1/\eps)^{o(1)}$ in the space-query product. First, we give a more detailed exposition of existing constructions.

All approximate distance oracles, including ours, use \emph{balanced shortest path separators}. The influential results of Lipton and Tarjan~\cite{LT79,LT80} showed that a triangulated planar graph of $n$ vertices has a separator $C$, called a balanced shortest path separator, which is a Jordan curve consisting of two shortest paths and one edge connecting the two endpoints of the paths, such that there are at most $\frac{2n}{3}$ vertices in the interior and exterior of $C$, denoted by $\inter(C)$ and $\exter(C)$, respectively.

A key idea in the construction of Thorup~\cite{Thorup04} and Klein~\cite{Klein02} is to use \emph{portals} along each shortest path of a shortest path separator. They showed that for every vertex $v$ in one side of $C$, say the interior, one can find a set $P_v$ of $1/\eps$ vertices, called portals, for $v$ such that for every vertex $u \in \exter(C)$, there exists a portal $p\in P_v$ where $d_G(v,p) + d_G(u,p) \leq (1+\eps) d_G(u,v)$. This means that for every vertex $v\in V$, we only need to store $O(1/\eps)$ distances to vertices in $P_v$ such that the distance between any two vertices $u,v$ in two different sides of $C$ can be approximated within $(1+\eps)$ factor by computing $\min_{p\in P_v, q\in P_u}\{d_G(v,p) + d_G(u,q)\}$ in time $O(|P_v| + |P_u|) = O(1/\eps)$. The distances between vertices in the same side of $C$ can be handled recursively, at the cost of a $\log(n)$ factor in the space bound since the depth of the recursion is $O(\log(n))$. 

We view the portalization of Thorup~\cite{Thorup04} and Klein~\cite{Klein02} as a \emph{local portalization} scheme in the sense that  each vertex needs its own set of portals. Evading $O(1/\eps^2)$ factor in the space-query product requires  breaking the locality. The follow-up construction of  Kawarabayashi, Klein, and Sommer used $r$-division~\cite{Federickson87} on top of the constructions by Thorup~\cite{Thorup04} and Klein~\cite{Klein02}; their goal is to have an oracle with $O(n)$ space and $O(\eps^{-2}\log^2 n)$ query time. Their construction also followed local portalization  and, hence, the space-query product remained the same.

Kawarabayashi, Sommer, and Thorup~\cite{KST13} (KST) improved the quadratic bound $O(1/\eps^2)$ in the space-query product by breaking the locality of portals entirely.  Specifically, up to a factor of $\polylog(1/\eps,\log^*n)$, they improved the space to $n\log(n)$ while keeping the same query time $1/\eps$.  Their approach  reduced constructing oracles with multiplicative stretch $(1+\eps)$ to constructing oracles with \emph{additive stretch} $+\eps D$ where $D$ is the diameter of the graph. (We say that an oracle has an additive stretch $+\Delta$ if for every two vertices $u$ and $v$, the distance returned by the oracle $\hat{d}(u,v)$ satisfies: $d_G(u,v)\leq \hat{d}(u,v)\leq d_G(u,v)  + \Delta$.) The reduction introduces a small loss of an $O(\log n)$ factor in the space and an $O(1)$ factor in query time. A key advantage of additive stretch over multiplicative stretch is that the portals become global:  for every shortest path (of length at most $D$) in a shortest path separator, we can place a set $P$ of $O(1/\eps)$ portals (independent of the vertices) such that for any two vertices $u,v$ in different sides of $C$, $\min_{p \in P}\{d_G(u,p)+d_G(v,p)\} \leq d_G(u,v) + \eps D$. Thus, vertices in the graph now share the same set of portals, which is the source of the space improvement. To answer a query, they need to iterate over $P$, which takes $O(|P| = O(1/\eps))$ time.

Subsequent constructions~\cite{CS19,GX19,LP19} improved the query time of the KST oracle to either $O(1)$~\cite{GX19,LP19}  or $O(\log(1/\eps))$~\cite{CS19} at the cost of a  large dependency on $1/\eps$ in the space.   Gu and Xu~\cite{GX19} employed the distance encoding argument of Weimann and Yuster~\cite{WY16} that has a factor $2^{O(1/\eps)}$ in the space.  Li and Parter~\cite{LP19} reduced the factor $2^{O(1/\eps)}$ to $O(\poly(1/\eps))$ using their VC-dimension technique.  Chan and Skrepetos~\cite{CS19} employed the Voronoi diagram technique of Cabello~\cite{Cabello18}; their construction broke away from the global portalization however: each vertex in the graph must store its own Voronoi diagram defined by its distances to the portals. As a result of this locality, the space of their oracle has a factor $(1/\eps)^4$. Thus, the Voronoi diagram technique, though a cornerstone of \emph{exact} distance oracle constructions (see \Cref{subsec:related-work}), does not seem to be the right approach for breaking the linear factor $1/\eps$ in the space-query product achieved by KST oracle.

Viewing KST oracle through our global lens of portalization is particularly illuminating for breaking the linear $1/\eps$ bound. Specifically, we show that, for breaking the $O(1/\eps)$ bound, it suffices to  store $\poly(1/\eps)$ machine words \emph{globally} per shortest path separator. More precisely, each time we apply the shortest path separator to separate a graph, we could store a global data structure of up to $\poly(1/\eps)$ words of space. Every vertex in the graph has a pointer to (a portion of) the data structure; the pointer only costs $O(\log(1/\eps))$ bits of space.  In retrospect, the KST  construction can be seen as storing only $O(1/\eps)$ words of space globally, one word for each portal of the paths in the separator. The key difference of our construction over KST's is that, in KST, each vertex needs to compute distances to the (shared) portals during the query stage, then computing distances between two vertices requires looping over the portals that takes $O(1/\eps)$ time. Our key idea to  remove the $1/\eps$ factor completely in the query time is the following. We precompute a small set of approximate distances in the global data structure. Then, given two vertices and their pointers to  the global data structure, we can look up their approximate distance in $O(1)$ time. 

In our construction, each vertex holds a pointer to an \emph{approximate distance pattern} stored in the global data structure in the graph. Our approximate distance pattern is an approximate version of the \emph{exact distance pattern} introduced by Fredslund-Hansen, Mozes, and Wulff-Nilsen~\cite{FMW20} in their construction of \emph{exact distance oracles} for \emph{unweighted} planar graphs. (Other than using distance patterns, their construction is different from ours and other approximate distance oracle constructions, and it only works for unweighted graphs.) An approximate distance pattern encodes the (approximate) distances from a vertex to the portals on a shortest path of the shortest path separator. We pre-compute all approximate distance patterns and the distance between any two  approximate distance patterns, and store this information in the global data structure. The idea is that, given  access to two approximate distance pointers of two vertices $u$ and $v$, we can access their pre-computed distance in $O(1)$ time in the global data structure. Our global data structure has only $\poly(1/\eps)$ space and, hence, the number of patterns must also be $\poly(1/\eps)$; for this, we employ the VC-dimension argument of Li and Parter~\cite{LP19}.

 A problem with using a global data structure of size $\poly(1/\eps)$ for each subgraph of $G$  arising in the construction is the space bound. Specifically, if we recursively decompose $G$ using balanced shortest path separators until each subgraph has $O(1)$ size, we end up with a separation tree, denoted by $\mathcal{T}$, with $O(n)$ nodes. As each node of $\mathcal{T}$ is associated with a global data structure of size $\poly(1/\eps)$ (for the subgraph of that node), the total size of the data structure is $O(n\poly(1/\eps))$. A possible solution to this problem is the following simple idea used by many oracle constructions~\cite{KKS11,KST13,CS19}: stop separating a subgraph once it has size $O(1/\eps^c)$ for some sufficiently big $c$. (For each subgraph of size $O(1/\eps^c)$, the standard approach is to use an exact distance oracle~\cite{KKS11,KST13}; we will come back to this issue later.)  The number of nodes of the separation tree is $O(n/\eps^{c})$, and hence the total size of all data structures associated with its nodes is  $O((n/\eps^{c}) \poly(1/\eps)) = O(n)$ for an appropriate choice of $c$. Then for each vertex $v$, we store a pointer to its approximate pattern w.r.t. the boundary portals of the leaf subgraph in $\mathcal{T}$ containing $v$.

 Yet a new problem arises: for each vertex $v$ in a leaf node $\alpha$ of $\mathcal{T}$, we need to compute the approximate pattern  of $v$   w.r.t. the portals of an ancestor  node, say $\beta$. Following Fredslund-Hansen, Mozes, and Wulff-Nilsen~\cite{FMW20}, we can compute a pattern \emph{ induced by} the approximate pattern of $v$ for portals of $\alpha$ (see \Cref{def:patter-induced} for a precise definition). The issue is that, since distances are approximate, the induced pattern might not be in the set of approximate patterns stored at $\beta$; therefore, we can no longer look up the distance stored at the global data structure of $\beta$. (In the exact distance setting of~\cite{FMW20}, this issue does not happen since the induced pattern  of an exact pattern is another exact pattern.)  Our key idea to overcome this problem is the following. We show that the induced pattern is close (in $\ell^{\infty}$ norm) to an approximate distance pattern stored at $\beta$. Then, for each induced pattern $\mathbf{p}$ between $\alpha$ and $\beta$, we store a \emph{pointer} to the approximate pattern in $\beta$ close to $\mathbf{p}$. Therefore, once the induced pattern is computed, we follow the pointer to the closest approximate pattern stored at $\beta$.

 We now go back to the subgraphs of size $O(1/\eps^c)$ associated with leaves of  $\mathcal{T}$. The standard approach is to use an exact distance oracle for each subgraph~\cite{KKS11,KST13}. For breaking the linear bound $1/\eps$ in the space-query product, it suffices to use the oracle of Charalampopoulos et al.~\cite{CGMW19}. Here we use the recent exact distance oracle by Long and Pettie~\cite{LP20} instead to obtain a better dependency on $1/\eps$. Specifically, we apply the Long-Pettie oracle (for $n$-vertex graphs) in two different regimes: (a) $n^{1+o(1)}$ space and $\log(n)^{2+o(1)}$ query time and (b) $n\log^{2+o(1)}(n)$ space and $n^{o(1)}$ query time. Two regimes lead to two different oracles with additive stretch $+\eps D$: regime (a) gives an oracle with $O(n(1/\eps)^{o(1)}\log(n))$ space and $O(\log^{2+o(1)}(1/\eps))$  query time and regime (b) gives an oracle with $O(n\log^{2+o(1)}(1/\eps)\log(n))$ space and $O((1/\eps)^{o(1)})$  query time. 
 
 Finally, with some additional ideas on top of our framework, we remove the $\log(n)$ factor in the space of the additive oracle by recursion at the cost of an additive $\log^*(n)$ term in the space. We note that, unlike KST, our oracle does not have the factor $\log^*(n)$ in the query time.

 \begin{restatable}{theorem}{AdditiveOracle}\label{thm:additiveOracle} Let $\eps > 0, D > 0$ be a positive parameter and $G = (V,E,w)$ be an undirected $n$-vertex planar graph of diameter $D$. There is an approximate distance oracle of additive stretch $+\eps D$ with construction time $O(n\poly(\log n,\eps))$ and: 
 	\begin{itemize}
 		\item[(1)]  $O(n((1/\eps)^{o(1)} + \log^*(n)))$ space  and  $\log^{2+o(1)}(1/\eps)$ query time or 
 		\item[(2)] $O(n(\log^{2+o(1)}(1/\eps) + \log^*(n)))$ space  and  $(1/\eps)^{o(1)}$ query time.
 	\end{itemize}
 \end{restatable}

\subsection{Related Work}\label{subsec:related-work}

A closely related data structure in planar graphs is \emph{exact} distance oracles. (All exact oracles mentioned in the following work for planar \emph{digraphs}; thus, they are naturally applicable to planar undirected graphs.) There is a very long line of work on constructing exact distance oracles,  starting from the seminal papers of Lipton and Tarjan~\cite{LT79,LT80} who constructed an exact oracle with $O(n^{3/2})$ space and $O(\sqrt{n})$ query time. Subsequent results~\cite{ACCDSZ96,Djidjev96,CX00,FR01,Cabello10,Wulff-Nilsen10,MS12} improved the result of Lipton and Tarjan in two ways: designing new space-query time trade-offs~\cite{ACCDSZ96,Djidjev96,CX00,Wulff-Nilsen10} or obtaining a truly subquadratic space-query product~\cite{Djidjev96,CX00,Cabello10,FR01,MS12}. However, none of these oracles has a \emph{truly subquadratic} space and \emph{polylogarithmic query time} until the work of Cohen-Addad, Dahlgaard, and WulffNilsen~\cite{CDW17}. Specifically, they constructed an oracle with $O(n^{5/3})$ space and $O(\log n)$ query time. The result of Cohen-Addad, Dahlgaard, and Wulff-Nilsen  is the major turning point for exact distance oracles: they were the first to use the Voronoi diagram technique of Cabello~\cite{Cabello18}.  Follow-up results~\cite{GMWW18,CGMW19,LP20}, all based on the Voronoi diagram technique, significantly improved the space-query time trade-off of Cohen-Addad, Dahlgaard, and Wulff-Nilsen~\cite{CDW17}, culminating in the oracles by Long and Pettie~\cite{LP20} that have (i) $O(n^{1+o(1)})$ space and $O(\log^{2+o(1)} n)$ query time or (ii) $O(n\log^{2+o(1)})$ space and $O(n^{o(1)})$ query time. (For various other trade-offs, see Table 1 in~\cite{LP20} for details.) The $n^{o(1)}$ factor, while sublinear, is $2^{\Omega(\sqrt{\log n})}$.  Therefore, though we have witnessed tremendous progress on exact distance oracles, the dependency on $n$ of spaces/query time of exact oracles remains far from that of approximate oracles.  

The exact oracle construction of Fredslund-Hansen, Mozes, and Wulff-Nilsen~\cite{FMW20} is fundamentally different from the constructions mentioned above: the main tool is the VC-dimension technique of Li and Parter~\cite{LP19}. Their main goal is to get an oracle with a constant query time; the space bound is $O(n^{5/3 + \delta})$ for any fixed constant $\delta > 0$. The space-query product of their oracle is not competitive with the Voronoi-diagram-based oracles~\cite{GMWW18,CGMW19,LP20}. Furthermore, their construction only works for undirected and unweighted planar graphs.

\section{Preliminaries}

Given a graph $G$, we denoted by $V(G)$ the vertex set of $G$ and $E(G)$ the edge set of $G$. We denote  an edge-weighted graph $G$ with a vertex set $V$, edge set $E$, and  non-negative edge-weight function $w: E\rightarrow \mathbb{R}^+$ by $G = (V,E,w)$. For any two vertices $u,v\in V$, we denote by $d_G(u,v)$ the distance between $u$ and $v$ in $G$. We denote by $SP(u,v,G)$ a shortest path from $u$ to $v$ in $G$. For a given path $P$ containing two vertices $x$ and $y$, we denote by $P[x,y]$ the $x$-to-$y$ subpath of $P$.

 Let $G = (V,E,w)$ be an edge-weighted planar graph equipped with a planar embedding, called a \emph{plane graph}. A region $R$ of $G = (V,E,w)$ is a subgraph of $G$.  A hole of $R$ is a face of $R$ that is not a face of $G$. The boundary of $R$, denoted by $\partial R$, is the set of vertices of $R$ that are on the boundaries of the holes of $R$. Vertices in $V(R)\setminus \partial R$ are called \emph{interior vertices}. A vertex $u\in V(G)\setminus V(R)$ is inside a hole $h$ if $u$ is embedded inside the face $h$ of $R$.

Next we define the notion of \emph{crossing} between two simple paths, say $P$ and $Q$, drawn on the plane.   We say that a path $X$ is a proper subpath of $P$ if $X$ does not contain any endpoint of $P$. Assume that there exists a maximal subpath $X\subseteq P\cap Q$ that is proper. We orient $P$ and $Q$ such that their orientations agree on $X$. Let $B_{\eps}$ be a topological disk containing all points on the plane of distance at most  infinitesimal $\eps > 0$ from points on $X$.  $P$ partitions   $B_{\eps}$ in two regions called the left side and the right side of $P$. We say that $Q$ crosses $P$ if the edge  entering $X$ and the edge leaving $X$ of $Q$ contain points in different sides of $P$ (see Figure~\ref{fig:crossing}(a)).   This crossing definition generalizes naturally to the case where $Q$ is a cycle instead of a path; in this case, any subpath of $Q$ is proper.

\begin{wrapfigure}{r}{0.3\textwidth}
	\vspace{-30pt}
	\begin{center}
		\includegraphics[width=0.28\textwidth]{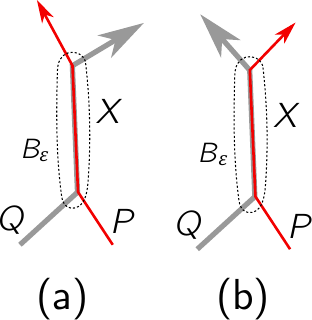}
	\end{center}
	\vspace{-15pt}
	\caption{(a) $P$ and $Q$ cross. (b) $P$ and $Q$ are non-crossing.}
	\label{fig:crossing}
	\vspace{-20pt}
	
\end{wrapfigure}

\paragraph{Distance preserving minors.} Let $P$ be a set of terminals in a  graph $G = (V,E,w)$. Let $K$ be the graph obtained by adding all pairwise shortest paths in $G$ between terminals in $P$, and contracting degree-2 vertices that are not in  $P$; we assume that the shortest paths are chosen in such a way that the intersection of any two shortest paths is either empty or connected. The weight of an edge in $K$ is the shortest distance between its endpoints in $G$. $K$ is a minor of $G$, and is called a \emph{distance preserving minor} for $P$~\cite{KNZ14}. If $G$ is planar, then $K$ is also planar. 

\begin{lemma}[Theorem 2.1~\cite{KNZ14}]\label{lm:dist-minor-size} Let $P$ be a set of $k$ terminal in a graph $G = (V,E,w)$, then its distance preserving minor $K$ has size $O(k^4)$.
\end{lemma}
Our construction uses \Cref{lm:dist-minor-size} for the case where $G$ is planar and $P$ is on the outer face of $G$.

\paragraph{Exact distance oracles.~} Long and Pettie~\cite{LP20} constructed an exact distance oracle for planar digraphs. In our paper, we use their results for undirected graphs.

\begin{theorem}[Theorem 1.1~\cite{LP20}]\label{thm:LongPettie} Let $G = (V,E,w)$ be any given  planar digraph with $n$ vertices. We can construct an exact distance oracle in time $n^{3/2+o(1)}$ that has:
	\begin{enumerate}
		\item[(1)]  $n^{1 + o(1)}$ space and $\log^{2+o(1)}n$ query time or
		\item[(2)] $n\log^{2+o(1)}$  space  and $n^{o(1)}$ query time.
	\end{enumerate}
\end{theorem}

\paragraph{VC dimension.~} Let $U$ be a ground set, and $\mathcal{F}$ be a family of subsets of $U$. We say that $\mathcal{F}$ shatters a set $X\subseteq U$ if for every subset $Y\subseteq X$, there exists a set $Z\in \mathcal{F}$ such that $Z\cap X = Y$. We say that $\mathcal{F}$ has VC-dimension $k$ if the largest set shattered by $\mathcal{F}$ has size $k$. The famous Sauer–Shelah lemma~\cite{Sauer72,Shelah72} bounds the size of $\mathcal{F}$ when its VC-dimension is at most $k$.

\begin{lemma}[Sauer–Shelah Lemma]\label{lm:SS}  Let $\mathcal{F}$ be a family of subsets of a ground set with $n$ elements. If  VC-dimension of $\mathcal{F}$ is at most $k$, then $|\mathcal{F}| = O(n^{k})$.
\end{lemma}

We use $[k]$ to denote the set $\{1,2\ldots, k\}$. If $\vec{x}$ is a $k$-dimensional vector, we denote by $\vec{x}[i:j]$ for given $i\leq j\leq k$ the $(j-i+1)$-dimensional vector such that $s$-th entry of $\vec{x}[i:j]$ is $\vec{x}[i+s-1]$ for any $s\in [j-i+1]$. We call vector $\vec{x}[i:j]$ the \emph{$(i,j)$-restriction} of $\vec{x}$.

Let $\mathbf{x},\mathbf{y}\in \mathbb{R}^k$ be  two $k$-dimensional vectors.  If $\lVert \mathbf{x},\mathbf{y}\rVert_{\infty} \leq \delta$, we write $\mathbf{x}\approx_{\smallfont{\delta}} \mathbf{y}$.  (The same notation applies to scalars since we can view them as $1$-dimensional vectors.) Observe by the triangle inequality that:

\begin{observation}\label{obs:apprx} If $\mathbf{x}\approx_{\delta_1} \mathbf{y}$ and $\mathbf{y}\approx_{\delta_2} \mathbf{z}$, then $\mathbf{x}\approx_{\delta_1+\delta_2} \mathbf{z}$.
\end{observation}

We can show directly from the definition that:
\begin{claim}\label{clm:min-apprx} If $\mathbf{x}_1\approx_{\smallfont{\delta_1}} \mathbf{x}_2$ and $\mathbf{y}_1\approx_{\smallfont{\delta_2}} \mathbf{y}_2$, then:
	\begin{equation*}
		\min_{i\in [k]}\{\mathbf{x}_1[i] + \mathbf{y}_1[i]\} \approx_{\smallfont{\delta_1 + \delta_2}}\min_{i\in [k]}\{\mathbf{x}_2[i] + \mathbf{y}_2[i]\}~.
	\end{equation*}
	where $k$ is the dimension of these vectors.
\end{claim}
\begin{proof} Observe by definition that $|\mathbf{x}_1[i] - \mathbf{x}_2[i]| \leq \delta_1$ and $|\mathbf{y}_1[i] - \mathbf{y}_2[i]| \leq \delta_2$ for every $i\in [k]$. It follows from the triangle inequality that  $ \mathbf{x}_1[i] +  \mathbf{y}_1[i]\approx_{\smallfont{\delta_1 + \delta_2}}\mathbf{x}_2[i] + \mathbf{y}_2[i]$ for every $i\in [k]$. Let $i^{*} = \argmin_{i\in [k]}\{\mathbf{x}_2[i] + \mathbf{y}_2[i]\}$. Then:
	\begin{equation*}
		\begin{split}
			\min_{i\in [k]}\{\mathbf{x}_1[i] + \mathbf{y}_1[i]\} &\leq  \mathbf{x}_1[i^*] +  \mathbf{y}_1[i^*] \\
			&\leq \mathbf{x}_2[i^*] + \mathbf{y}_2[i^*] + (\delta_1+\delta_2) \\&= \min_{i\in [k]}\{\mathbf{x}_2[i] + \mathbf{y}_2[i]\} + (\delta_1+\delta_2) ~.
		\end{split}
	\end{equation*}
	By the same argument, $\min_{i\in [k]}\{\mathbf{x}_2[i] + \mathbf{y}_2[i]\} \leq \min_{i\in [k]}\{\mathbf{x}_1[i] + \mathbf{y}_1[i]\} + (\delta_1+\delta_2)$; this implies the claim.
\end{proof}

\section{Distance Oracles with Additive Stretch}

In this section, we construct an oracle with additive stretch for planar graphs as claimed in \Cref{thm:additiveOracle}. For a simpler presentation of the ideas, we first prove a weaker version of \Cref{thm:additiveOracle}, which is \Cref{thm:additiveOracleEasy} below, where we allow a $\log^{o(1)}(n)$ factor in the space and query time. We then present a full proof of \Cref{thm:additiveOracle} in \Cref{subsec:additiveStrong}. 

\begin{theorem}\label{thm:additiveOracleEasy} Let $\eps > 0, D > 0$ be a positive parameters and $G = (V,E,w)$ be an undirected $n$-vertex planar graph of diameter $D$. There is an approximate distance oracle with   additive stretch $+\eps D$ that has $O(n\poly(\log n,\eps))$ construction time  and:
	\begin{itemize}
		\item[(1)]  $n(1/\eps)^{o(1)}\log^{o(1)}(n)$ space  and  $\log^2(1/\eps) + (\log\log(n))^{2+o(1)}$ query time or 
		\item[(2)] $n(\log^{2+o(1)}(\log n) + \log^{2+o(1)}(1/\eps))$ space  and  $\log^{o(1)}(n)(1/\eps)^{o(1)}$ query time.
	\end{itemize}
\end{theorem}

This section is organized as follows. In \Cref{subsec:approxPattern} we introduce the notion of approximate patterns. In \Cref{subsec:approxDecoding}, we show how to compute an approximate distance of two vertices given their approximate patterns to portals on a shortest path separator. In  \Cref{subsec:composition}, we study the composition of two approximate distance patterns, and show that the composition is close to another approximate distance pattern in $\ell^{\infty}$ norm.  In \Cref{subsec:easyOracle} we prove \Cref{thm:additiveOracleEasy}  and in \Cref{subsec:additiveStrong}, we extend the proof of \Cref{thm:additiveOracleEasy} to \Cref{thm:additiveOracle}.

\subsection{Approximate Patterns~}\label{subsec:approxPattern}

Let $G = (V,E,w)$ be a planar graph and $ \sigma$ be a sequence of $k$ vertices of $G$; the $i$-th vertex of $\sigma$ is denoted by $\sigma_i$. For a real number $\Delta\in \mathbb{R}$, which could be negative, positive or zero, and an index $i \in [k-1]$, we define:
\begin{equation}\label{eq:distance-index}
	A_i^{\Delta} = \{v\in V(G)| d_G(v,\sigma_{i+1})-d_G(v,\sigma_i)\leq \Delta\}~.
\end{equation} 
We call the pair $(i,\Delta)$ a \emph{distance index} and $A_i^{\Delta}$ a vertex set associated with the distance index $(i,\Delta)$. See \Cref{fig:VC-dim-example}.  The following theorem is proved by Li and Parter~\cite{LP19}. 

\begin{theorem}[Theorem 3.7, Li-Parter~\cite{LP19}]\label{thm:LP19} Let $G = (V,E,w)$ be a planar graph and $\sigma$ be a sequence of $k$ vertices in clockwise order on a face $f$  of $G$. Let $M$ be a set of real numbers. For each vertex $v\in G$, let $X_v = \{(i,\Delta): i\in [k-1], \Delta \in M, v\in A^{\Delta}_{i}\}$ be the set of distance indices whose associated vertex sets contain $v$.  Let $\mathcal{F} = \{X_v\}_{v\in V(G)}$ be a family of sets of distance indices. Then $\mathcal{F}$ has VC-dimension at most $3$.
\end{theorem}

\begin{wrapfigure}{r}{0.4\textwidth}
	\vspace{-20pt}
	\begin{center}
		\includegraphics[width=0.38\textwidth]{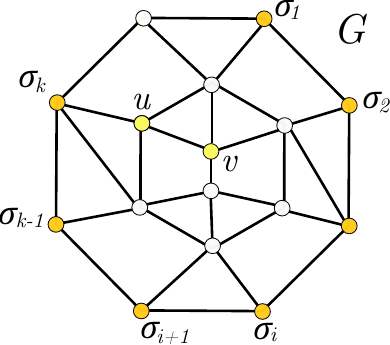}
	\end{center}
	\caption{$G$ is unweighted and $\sigma$ is on the outer face. Both $u$ and $v$ are in $A^{0}_i$.}
	\label{fig:VC-dim-example}
\end{wrapfigure}

\begin{remark}\label{rm:LP-understanding} An intuitive interpretation of \Cref{thm:LP19} is the following. The (typically finite) set $M$ tells us the difference between the distances from a vertex $v$ to two consecutive vertices in the sequence $\sigma$:  if $d_G(v,\sigma_{i+1})-d_G(v,\sigma_i) \leq \Delta$ for some $\Delta \in M$, then $(i,\Delta) \in X_v$. For each $i$, let $\Delta^*_{i}\in M$ be the smallest such that  $(i,\Delta^*_{i}) \in X_v$. Then, given $d_G(v,\sigma_1)$  and $X_v$, we can inductively recover an \emph{upper bound} on $d_G(v,\sigma_{i+1})$ for any given $i\in [1,k-1]$ by unrolling the recursion  $d_G(v,\sigma_{i+1})\leq d_G(v,\sigma_i) + \Delta^*_{i}$. That is, we get $d_G(v,\sigma_{i+1}) \leq d_G(v,\sigma_1) + \sum_{j=1}^{i}\Delta^*_{j}$.  Depending on the choice of $M$, this upper bound could be an exact or approximate estimation of the distance  $d_G(v,\sigma_{i+1})$. Thus, $X_v$ and $M$ encode the distance information from $v$ to vertices in $\sigma$; the notion of approximate distance encoding below formalizes this intuition. From this point of view, the family $\mathcal{F}$ contains the approximate distance encodings of all vertices of $G$ to vertices in $\sigma$. By \Cref{lm:SS}, \Cref{thm:LP19} implies that  there are only a polynomial number (in $k$ and $|M|$) of approximate distance encodings; the number of encodings does not depend on $n$!  
\end{remark}

We note that Theorem 3.7 in~\cite{LP19} is only stated for $\Delta\in \{-1,0\}$; however, as noted by Li and Parter~\cite{LP19} in the proof of Theorem 3.7, it holds for any set of real numbers. They used the general version to approximate weighted diameters of planar graphs. 

Our construction relies on the notion of approximate patterns.  For a given positive real number $\delta$, we define  $\appr{\delta}{a}$ to be the closest integer multiple of $\delta$ that is at least $a$. Specifically,
\begin{equation}\label{eq:delta-approx}
	\appr{\delta}{a} = \lceil \frac{a}{\delta} \rceil \cdot \delta \qquad \forall a \in \mathbb{R}~.
\end{equation}

Next, we define approximate pattern and approximate distance decoding. Our approximate pattern is  the approximate version of (exact) patterns  introduced  by Fredslund-Hansen, Mozes, and Wulff-Nilsen~\cite{FMW20}. 

\begin{definition}[Approximate Pattern and Distance Encoding]\label{def:appx-pattern} Let  $\sigma $ be a sequence of vertices in a graph $G$. Let $u$ be a vertex in $G$.  A \emph{$\delta$-approximate pattern} of $u$ w.r.t. $\sigma$ in $G$ for some parameter $\delta > 0 $ is a $(k-1)$-dimensional  vector  $\vec{p}$ such that $\vec{p}[i] = \appr{\delta}{d_G(u,\sigma_{i+1})-d_G(u,\sigma_i)}$ for every $i\in [k-1]$. \\	 
	A \emph{$\delta$-approximate distance encoding} of $u$ w.r.t.  $\sigma$ in $G$ is a $k$-dimensional vector $\vec{d}$ such that $\vec{d}[1] =  \appr{\delta}{d_{G}(u,\sigma_1)}$ and $\vec{d}[2:k] = \vec{p}$. That is, $\vec{d}[i] =  \vec{p}[i-1] = \appr{\delta}{d_G(u,\sigma_i)-d_G(u,\sigma_{i-1})}$ for all $2\leq i\leq k$.
\end{definition}

Given the distance encoding $\vec{d}$ of a vertex $u$, we can \emph{decode} $\vec{d}$ to get a $k$-dimensional distance vector $\vec{a}$ from $u$ to vertices in $\sigma$ by computing $\vec{a}[i] = \sum_{j=1}^{i}\vec{d}[j]$ for every $i\in [k]$. (In \Cref{lm:approx-dist} below, we show that  $\vec{a}[i]$ is close to $d_G(u,\sigma_i)$.) We can generalize the decoding procedure for any $k$-dimensional vector $\vec{x}$ even if it does not correspond to a distance encoding of any vertex in the graph.  We will use this generalization in our distance oracle construction.

\begin{definition}[Distance Decoding]\label{def:decode} Let $\vec{x}$ be a  $k$-dimensional vector. The distance decoding of $\vec{x}$  is a $k$-dimensional vector, denoted by $\vec{x}^{\leq}$, such that $\vec{x}^{\leq}[i] = \sum_{j=1}^{i}\vec{x}[j]$.
\end{definition}

Next, we show that the distance decoding of a $\delta$-approximate distance encoding of a vertex $u$ in the graph gives approximate distances from $u$ to vertices in $\sigma$. 

\begin{lemma}\label{lm:approx-dist} Let $\vec{d}$  be a $\delta$-approximate distance encoding of $u$ w.r.t. a sequence $\sigma$ of $k$ vertices in graph $G$. For any $i\in [k]$, we define $\tilde{d}_G(u,\sigma_i)  = \vec{d}^{\leq}[i]$. Then,  it holds that
	\begin{equation}\label{eq:approx-dist}
		d_G(u,\sigma_i) \leq \tilde{d}_G(u,\sigma_i) \leq d_G(u,\sigma_i) + i\cdot\delta.~
	\end{equation}
\end{lemma}
\begin{proof} Let $\vec{p}$ be the $\delta$-approximate pattern of $u$ w.r.t. $\sigma$.  Observe by the definition of  $\appr{\delta}{\cdot}$ and  \Cref{def:appx-pattern} that for any $j\in [k-1]$:
	\begin{equation}\label{eq:expand}
		d_G(u,\sigma_{j+1})-d_G(u,\sigma_{j})\leq \vec{p}[j] \leq d_G(u,\sigma_{j+1})-d_G(u,\sigma_{j}) + \delta
	\end{equation}
	By summing both sides of \Cref{eq:expand} when $j = 1,\ldots, i-1$, it follows that:
	\begin{equation*}
		d_G(u,\sigma_i)-d_G(u,\sigma_1)\leq\sum_{j=1}^{i-1} \vec{p}[j] \leq d_G(u,\sigma_i)-d_G(u,\sigma_1) + (i-1)\delta~.
	\end{equation*}
	Thus, we have:
	\begin{equation}\label{eq:semi-final}
		d_G(u,\sigma_i)\leq d_G(u,\sigma_1) + \sum_{j=1}^{i-1}\vec{p}[j] \leq d_G(u,\sigma_i)+ (i-1)\delta~.
	\end{equation}
	By definition of  $\appr{\delta}{\cdot}$, $ d_G(u,\sigma_1) \leq \appr{\delta}{ d_G(u,\sigma_1)}\leq  d_G(u,\sigma_1) + \delta$. Thus, \Cref{eq:semi-final} implies:
	\begin{equation}\label{eq:final}
		d_G(u,\sigma_i)\leq \appr{\delta}{ d_G(u,\sigma_1)} + \sum_{j=1}^{i-1}\vec{p}[j] \leq d_G(u,\sigma_i)+ i\delta~.
	\end{equation}
 By \Cref{def:appx-pattern}, $\appr{\delta}{ d_G(u,\sigma_1)} + \sum_{j=1}^{i-1}\vec{p}[j] = \sum_{j=1}^{i} \vec{d}[j] = \vec{d}^{\leq}[i]$. The lemma now follows from \Cref{eq:final}.
\end{proof}

We note that by the definition of $\tilde{d}_G$ in \Cref{lm:approx-dist}, 
\begin{remark}\label{remrk:dusigma1} $\tilde{d}_G(u,\sigma_1) = \appr{\delta}{d_{G}(u,\sigma_1)}$.
\end{remark}

The following lemma bounds the number of patterns when $|d_G(u,\sigma_{i+1}) - d(u,\sigma_i)|$ is not much larger than $\delta$. 

\begin{lemma}\label{lm:pattern-bound} Let $G = (V,E,w)$ be a planar graph and $\sigma$ be a sequence of $k$ vertices in clockwise order on a face $f$  of $G$. Let $g$ be a non-negative integer such that:
	\begin{equation}\label{eq:usiPlus-usi}
		-g\delta \leq d_G(u,\sigma_{i+1}) - d(u,\sigma_i) \leq g\delta \qquad \forall i \in [k-1]
	\end{equation}
For every vertex $u\in V$, let $\vec{p}_u$ be the $\delta$-approximate pattern of $u$ w.r.t. $\sigma$. Let $\vec{P} = \{\vec{p}_u: u\in V\}$ be the set of all $\delta$-approximate patterns w.r.t. $\sigma$. Then $|\vec{P}| = O((kg)^3)$.
\end{lemma}
\begin{proof} Let $M = \{-g\delta, (-g+1)\delta, \ldots, -\delta, 0, \delta, \ldots, g\delta\}$ be a set of $(2g+1)$ real numbers. Recall that $X_u = \{(i,\Delta)\}$ is the set of distance indices associated with a vertex $u$, where $\Delta \in M$. By \Cref{thm:LP19}, $\mathcal{F}= \{X_v\}_{v\in V(G)}$ has VC-dimension at most $3$. Since the ground set $\{(i,\Delta)\}_{i\in [k-1], \Delta \in M}$ has size at most $k(2g+1)$, by \Cref{lm:SS}, $|\mathcal{F}| = O((k(2g+1))^3) = O((kg)^3)$.
	
	We show below that the map $\varphi$ that maps $\vec{p}_u$ to $X_u$ is a bijection from $\vec{P}$ to $\mathcal{F}$, which would imply the claimed bound on  $|\vec{P}|$. By definition, $\varphi$ is surjective. Next, we show that $\varphi$ is injective.
	
	Let $u\not=v$ be two vertices such that $\vec{p}_u\not=\vec{p}_v$. Then there exists some $i\in [k-1]$ such that $\appr{\delta}{d_G(u,\sigma_{i+1}) - d_G(u,\sigma_i)} \not= \appr{\delta}{d_G(v,\sigma_{i+1}) - d_G(v,\sigma_i)}$. Let $\Delta_u = \appr{\delta}{d_G(u,\sigma_{i+1}) - d_G(u,\sigma_i)}$ and $\Delta_v = \appr{\delta}{d_G(v,\sigma_{i+1}) - d_G(v,\sigma_i)}$. Since $\Delta_u\not=\Delta_v$, w.l.o.g, we assume that $\Delta_u < \Delta_v$. By the definition of  $\appr{\delta}{\cdot}$, $d_G(u,\sigma_{i+1}) - d_G(u,\sigma_i) \leq \Delta_u$. Thus, $u\in A^{\Delta_u}_i$, and that $(i,\Delta_u) \in X_u$. 
	
	Similarly, $v\in A^{\Delta_v}_i$. Furthermore, by the definition of $\appr{\delta}{\cdot}$ that  $\Delta_v$ is the least multiple of $\delta$ that is at least $d_G(v,\sigma_{i+1}) - d_G(v,\sigma_i)$. Thus,  there is no other $\Delta'\in M$ such that $\Delta' < \Delta_v$ and $v\in A^{\Delta'}_i$. Since $\Delta_u < \Delta_v$, $v\not\in  A^{\Delta_u}_i$. That is, $(i,\Delta_u)\not\in X_v$. It follows that $X_u\not=X_v$. Therefore,  $\varphi$ is injective.
\end{proof}

We conclude this section by defining the distance between an approximate pattern and a vertex. Recall that $\tilde{d}_G(u,\sigma_i)$ is the approximate distance from $u$ to vertex $\sigma_i$ computed from the  $\delta$-approximate distance encoding of $u$ w.r.t. $\sigma$ (see \Cref{lm:approx-dist}). 

\begin{definition}[Pattern-Vertex Distance]\label{def:patter-ver-dist} Let $v$ be a vertex and $\vec{p}$ be an approximate pattern (of some vertex) w.r.t. a sequence $\sigma$ in $G= (V,E,w)$. Then, the distance between  $\vec{p}$ and $v$ is defined as $\tilde{d}_G(v,\vec{p}) = \min_{1\leq i\leq k}\{\tilde{d}_G(v,\sigma_i) + \sum_{j=1}^{i-1}\vec{p}[j]\}$. 
\end{definition}

\subsection{Computing Distances from Approximate Distance Encodings}\label{subsec:approxDecoding}

A basic building block in our distance oracle is to query the distance between two vertices separated by a cycle $C$ given their approximate patterns to a sequence of vertices on the cycle.  In the work of  Fredslund-Hansen, Mozes, and Wulff-Nilsen~\cite{FMW20}, (exact) patterns are defined w.r.t. \emph{all vertices} on $C$. Using this fact, they can show that, to query the distance between a vertex $u$ outside a cycle $C$ to a vertex $v$, it suffices to compute the distance from $u$ to a fixed vertex $\sigma_1 \in C$, and compute the distance from $v$ to the pattern of $u$ w.r.t. $C$. Cycle $C$ does not have to have any special structure for the distance query to work, and this follows from the fact that their exact patterns encode exact distances. However, this property no longer holds in our setting, as pattern are approximate and defined w.r.t. a \emph{subset} of vertices on $C$.

We introduce a property, call the \emph{single-crossing property}, and we show that if the separating cycle $C$ is single-crossing, one can retrieve the approximate distance between $u$ and $v$ in different sides of $C$ based on their approximate patterns to the boundary. 

\begin{definition}[Single-Crossing Property] Let $C$ be a simple cycle in a plane graph $G = (V,E,w)$. We say that $C$ is single-crossing if for any two vertices $u,v$ such that $u \in \overline{\inter(C)}$ and $v\in \overline{\exter(C)}$, then there exists a shortest path from $u$ to $v$ crossing $C$ at most once. 
\end{definition}

Let $\sigma$ be a sequence of vertices ordered clockwise on $C$. We say that $\sigma$ is a \emph{$\tau$-cover} of $C$ for some $\tau \geq 0$ if for every $u\in C$, $d_C(\sigma_i,u)\leq \tau$ where $i \in [k]$ is the index such that $u\in C[\sigma_i,\sigma_{i+1}]$; here $\sigma_{k+1} = \sigma_1$.  In this section, we show the following. 

\begin{figure}[!htb]
	\includegraphics[width=.9\textwidth]{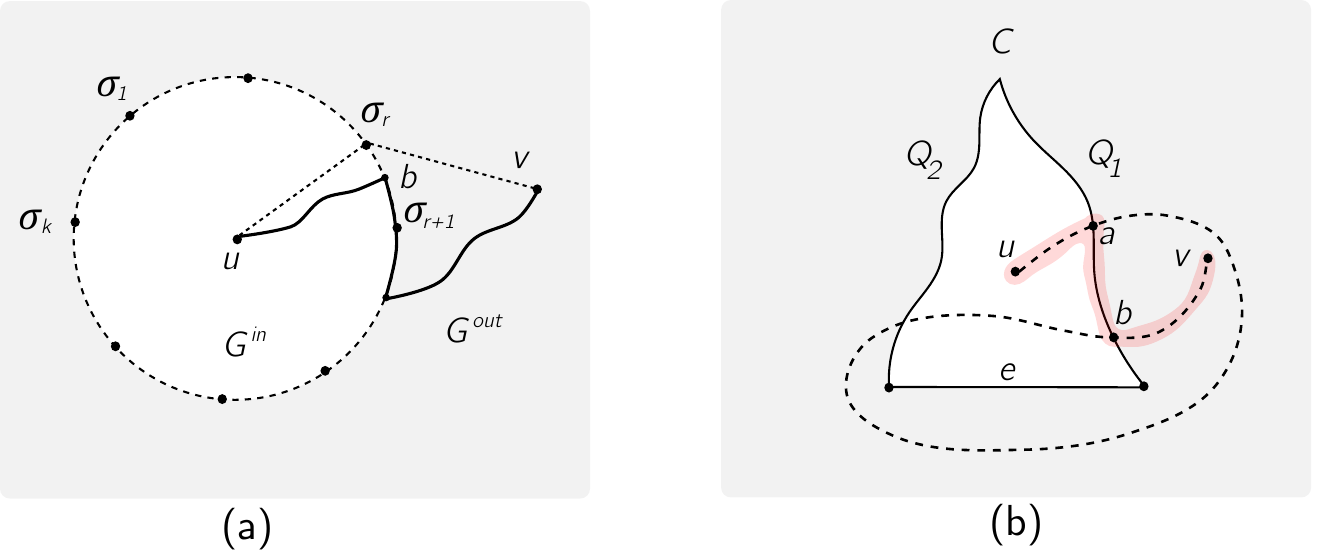}
	\caption{(a) The solid path from $u$ to $v$ is $SP(u,v,G)$; this path crosses the shortest path separator (dashed cycle) once. (b) The dashed path is a shortest path from $u$ to $v$ that crosses the path $Q_1$ in a shortest path separator at least twice. The red-highlighted path is a new shortest path from $u$ to $v$ that crosses $Q_1$ only once.}
	\label{fig:distEncode}
\end{figure}

\begin{lemma}\label{lm:separating-cycle}Let $\delta > 0, \tau \geq 0$ be parameters. Let  $C$ be single-crossing simple cycle of a plane graph $G = (V,E,w)$, and  $\sigma$ be a sequence of $k$ vertices ordered clockwise on $C$ that $\tau$-covers $C$. Let $G^{in}$ ($G^{out}$)  be the subgraph of $G$ induced by vertices inside (outside) or on $C$. Let $u\in V(G^{in})$ and $v\in V(G^{out})$ be any two vertices.  Let $\vec{p}_{u}$ and  $\vec{d}_{u}$ ($\vec{p}_{v}$ and  $\vec{d}_{v}$) be the $\delta$-approximate pattern  and $\delta$-approximate distance encoding, respectively, of $u$ ($v$)  w.r.t. $\sigma$ in $G^{in}$ ($G^{out}$). Let
	\begin{equation}\label{eq:dtilde-uv}
		\tilde{d}_G(u,v) \defi  \min_{1\leq i\leq k}\left\{\tilde{d}_{G^{in}}(u,\sigma_i) + \tilde{d}_{G^{out}}(v,\sigma_i)\right\} 
	\end{equation}
	where $\tilde{d}_{G^{in}}(u,\sigma_i) = \vec{d}^{\leq}_{u}[i]$ and $\tilde{d}_{G^{out}}(v,\sigma_i) = \vec{d}^{\leq}_{v}[i]$. Then, the followings hold:
	\begin{itemize}
		\item[(1)] $d_G(u,v)\leq \tilde{d}_G(u,v)\leq d_G(u,v) + 2k\delta + 2\tau$.
		\item[(2)] $\tilde{d}_G(u,v) = \tilde{d}_{G^{out}}(\vec{p}_{u}, v) + \tilde{d}_{G^{in}}(u,\sigma_1)$.
	\end{itemize}
\end{lemma}
\begin{proof}  See \Cref{fig:distEncode}(a) for an illustration. Note that $G$ is undirected. By \Cref{lm:approx-dist}, we have that $\tilde{d}_{G^{in}}(u,\sigma_i) \geq d_{G^{in}}(u,\sigma_i)$ and that $\tilde{d}_{G^{out}}(v,\sigma_i)\geq d_{G^{out}}(v,\sigma_i)$. By the triangle inequality, we have that $\tilde{d}_{G^{in}}(u,\sigma_i) + \tilde{d}_{G^{out}}(v,\sigma_i)\geq d_G(u,v)$ for any $i\in [k]$.  It follows that
	\begin{equation*}
		\tilde{d}_G(u,v) = \min_{i\in [k]}\{ d_{G^{in}}(u,\sigma_i) + d_{G^{out}}(v,\sigma_i)\}\geq d_G(u,v)~.
	\end{equation*}	
	This implies the left inequality in Item (1).
	
	To show the right inequality in Item (1), we observe that, since $C$  is simple-crossing, there is a shortest path $SP(u,v,G)$ from $u$ to $v$ cross $C$ at most once. Thus, there exists a vertex $b\in C$ such that $SP(u,v,G) = SP(u,b,G^{in})\circ SP(b,v,G^{out})$, where $SP(u,b,G^{in})$ ($SP(b,v,G^{out})$) is the shortest $u$-to-$b$ ($b$-to-$v$) path in $G^{in}$ ($G^{out}$).  Let $r\in [k]$ be such that $d_C[\sigma_{r},b]\leq \tau$; $r$ exists since $\sigma$ is a $\tau$-cover of $C$.By the triangle inequality we have:
	\begin{equation}\label{eq:dbsr}
		d_{G^{in}}(u, \sigma_{r}) + d_{G^{out}}(\sigma_{r}, v) \leq d_G(u,v)  + 2 d_C[\sigma_{r},b] \leq  d_G(u,v) + 2\tau
	\end{equation}
	Furthermore, by \Cref{lm:approx-dist}, the RHS of \Cref{eq:dtilde-uv} is at most
	\begin{equation*}
		\begin{split}
			\min_{i\in [k]}\{ d_{G^{in}}(u,\sigma_i) + d_{G^{out}}(v,\sigma_i) + 2i\delta\} &\leq d_{G^{in}}(u,\sigma_{r}) + d_{G^{out}}(v,\sigma_{r}) +  2k\delta\\
			&\leq d_G(u,v) + 2k\delta +  2\tau
		\end{split}
	\end{equation*}	
	by \Cref{eq:dbsr}. 
	
	Next, we prove Item (2). Note by \Cref{remrk:dusigma1} that $\tilde{d}_{G^{in}}(u,\sigma_1) = \appr{\delta}{d_{G^{in}}(u,\sigma_1)}$. By \Cref{def:appx-pattern} and \Cref{def:decode}, we have:
	\begin{equation}\label{eq:dist-u-sigmai-Gin}
		\tilde{d}_{G^{in}}(u,\sigma_i) = \vec{d}_u^{\leq}[i] =  \appr{\delta}{d_{G^{in}}(u,\sigma_1)} + \sum_{j=1}^{i-1}\vec{p}_{u}[j]
	\end{equation}
	  By \Cref{def:patter-ver-dist}, we have:
	\begin{equation}
		\begin{split}
			\tilde{d}_{G^{out}}(\vec{p}_{u}, v) + \tilde{d}_{G^{in}}(u,\sigma_1)  &=  \min_{1\leq i\leq k}\{\tilde{d}_{G^{out}}(v,\sigma_i) + \sum_{j=1}^{i-1}\vec{p}_{u}[j]\}  +  \tilde{d}_{G^{in}}(u,\sigma_1)  \\
			&=  \min_{1\leq i\leq k}\{\tilde{d}_{G^{out}}(v,\sigma_i) + \sum_{j=1}^{i-1}\vec{p}_{u}[j] + \appr{\delta}{d_{G^{in}}(u,\sigma_1)} \}\\
			&= \min_{1\leq i\leq k}\{\tilde{d}_{G^{out}}(v,\sigma_i) +  \tilde{d}_{G^{in}}(u,\sigma_i)\} \qquad \mbox{(by \Cref{eq:dist-u-sigmai-Gin})}\\
			&=  \tilde{d}_G(u,v)~,
		\end{split}
	\end{equation}
	as desired.
\end{proof}

By \Cref{lm:separating-cycle}, to obtain an approximate distance between $u$  and $v$, it suffices to know the $\delta$-approximate distance encodings of $u$ and $v$ w.r.t. $\sigma$ in $G^{in}$ and $G^{out}$, respectively.  We show later that, by choosing $\delta$ appropriately, the size of the set of all approximate distance encodings $\{\vec{d}_{u}\}_{u\in V(G^{in})}$ of $G^{in}$ is  polynomial in $1/\eps$. Thus, we can store all these distance encodings in a table, say $T_1$. Then for each vertex  $u \in G^{in}$, we only store a \emph{pointer} to  the corresponding entry in the table, which costs only $O(1)$ machine words (instead of $O(1/\eps)$ machine words to store the actual distance encoding of $u$).  The same holds for $v$ in graph $G^{out}$. Furthermore, we precompute distances $\tilde{d}$ from every pair of distance encodings, one from $G^{in}$ and the other from $G^{out}$, and store the result in a lookup table, say $T_2$, which costs only $O(\poly(1/\eps))$ space. To retrieve the distance from $u$ to $v$, we simply follow the pointers to access their approximate distance encodings in $T_1$ and then access the $\tilde{d}_G(u,v)$ in table $T_2$ in $O(1)$ time.  The  total space is $O(n + \poly(1/\eps))$. Note that this is only for querying distances from a vertex in $G^{in}$ to a vertex in $G^{out}$; we need to recurse to construct a data structure for all pairs of vertices. Furthermore, we want the space bound to be $\tilde{O}(n(1/\eps)^{o(1)})$ rather than $O(n + \poly(1/\eps))$, and for this, we need additional ideas.  Nevertheless, the fact that the space dependency on $1/\eps$ is additive instead of multiplicative is the key to our construction later. 

In our oracle construction, sometimes for a given vertex $u\in G^{in}$ and $v \in G^{out}$, we can only obtain approximate distance encodings that are sufficiently close to the true approximate distance encodings of $u$ and $v$. We show below that we can still recover the approximate distance between $u$ and $v$. To formally state our result, we need some notation. Given two approximate distance encodings $\vec{d}_1$ and $\vec{d}_2$ of dimension $k$, we say that:
\begin{equation}\label{eq:encoding-close}
	\vec{d}_1\approx_{\smallfont{\delta_1},\smallfont{\delta_2}}\vec{d}_2 \qquad \mbox{if } \vec{d}_1[1]\approx_{\smallfont{\delta_1}} \vec{d}_2[1] \mbox{ and } \vec{d}_1[2:k]\approx_{\smallfont{\delta_2}}\vec{d}[2:k]
\end{equation}

\begin{definition}[Approximate Distance from Approximate Distance Encodings]\label{def:encoding-dist} Let $\vec{d}_1$ and $\vec{d}_1$ be two approximate distance encodings of dimension $k$. We define their distance, denoted by, $\norm{\vec{d}_1, \vec{d}_2}$, as follows:
	\begin{equation}\label{eq:coding-norm}
		\norm{\vec{d}_1, \vec{d}_2}  = \min_{1\leq i\leq k}\{ \vec{d}_1^{\leq}[i] + \vec{d}_2^{\leq}[i]\}
	\end{equation}
\end{definition}

If $\vec{d}_u$ and $\vec{d}_v$ are two $\delta$-approximate distance encodings as in \Cref{lm:separating-cycle}, then $\norm{\vec{d}_u, \vec{d}_v} =  \tilde{d}_G(u,v)$ by  definition (\Cref{eq:dtilde-uv}). We show below that we can recover the approximate distance from $u$ to $v$ from the approximations of their approximate distance encodings. 

\begin{lemma}\label{lm:dist-approximation-enco} Let $\vec{d}_1$ and $\vec{d}_2$ be two approximate distance encodings of dimension $k$, and $u \in G^{in}$ and $v\in G^{out}$ be two vertices of $G$ where $G^{in}$ and $G^{out}$ are as defined in \Cref{lm:separating-cycle}. Let $\vec{d}_u$ ($\vec{d}_{v}$) be the $\delta$-approximate distance encoding of $u$ ($v$)  w.r.t. $\sigma$  in $G^{in}$ ($G^{out}$). Suppose that:
	\begin{equation*}
		\vec{d}_1 \approx_{\smallfont{\delta_1},\smallfont{\delta_2}}  \vec{d}_u \qquad \mbox{and} \qquad\vec{d}_2 \approx_{\smallfont{\delta_1},\smallfont{\delta_2}} \vec{d}_v
	\end{equation*}
Then, $\norm{\vec{d}_1,\vec{d}_2}  \approx_{\smallfont{2\delta_1 + 2(k-1)\delta_2}} \tilde{d}_G(u,v)$.
\end{lemma} 
\begin{proof} By \Cref{obs:apprx}, $\sum_{j=2}^{i} \vec{d}_1[j] \approx_{\smallfont{(i-1)\delta_2}} \sum_{j=2}^{i} \vec{d}_u[j]$ for any $i \in [k]$. Thus, by \Cref{def:decode} and \Cref{obs:apprx}, $\vec{d}_1^{\leq}[i]\approx_{\smallfont{\delta_1 + (i-1)\delta_2}} \vec{d}_u^{\leq}[i] = \tilde{d}_{G^{in}}(u,\sigma_i)$. It follows that $\vec{d}_1^{\leq} \approx_{\smallfont{\delta_1 + (k-1)\delta_2}}\tilde{d}_{G^{in}}(u,\cdot)$; here $\tilde{d}_{G^{in}}(u,\cdot)$ is the $k$-dimensional vector whose $i$-th component is $\tilde{d}_{G^{in}}(u,\sigma_i)$. By the same argument, we have that $\vec{d}_2^{\leq}\approx_{\smallfont{\delta_1 + (k-1)\delta_2}}\tilde{d}_{G^{out}}(v,\cdot)$. The lemma then follows from \Cref{clm:min-apprx} and \Cref{eq:dtilde-uv}.
\end{proof}

We close this section by showing that shortest path separators in planar graphs are single-crossing cycles. We rely on the well-known property that each shortest path separator cycle consists of two shortest paths and a single edge. 

\begin{lemma}\label{lm:single-crossing-sep}If a simple cycle $C$ of a plane graph $G(V,E,w)$ is composed of two shortest paths and a single edge, then $C$ is single-crossing.
\end{lemma}
\begin{proof} Suppose that $C = Q_1\circ e\circ Q_2$ where $Q_1,Q_2$ are shortest paths of $G$, and $e$ is a single edge. Let $u\in G^{in}$ and $v\in G^{out}$ be two vertices separated by $C$, where $G^{in}$ ($G^{out}$)  be the subgraph of $G$ induced by vertices inside (outside) or on $C$. Let $SP(u,v,G)$ be a shortest path from $u$ to $v$ in $G$ that crosses $C$ a minimum number of times. (If $u$ and $v$ are both on $C$, we regard $u$ as inside $C$ and $v$ as outside $C$.) By Jordan curve theorem, $SP(u,v,G)$ must cross $C$ by an odd number of times. If $SP(u,v,G)$ crosses $C$ at least 3 times, then $SP(u,v,G)$ must cross, say $Q_1$, at least twice. See \Cref{fig:distEncode}(b) for an illustration. Let $a$ and $b$ be the first and the last crossing points on the path $SP(u,v,G)$ oriented from $u$ to $v$. By replacing the subpath $SP(u,v,G)[a,b]$ by $Q_1[a,b]$, we obtain another path from $u$ to $v$ while the number of crossing is reduced by at least one, contradicting that $SP(u,v,G)$ has a minimum number of crossings. Thus, $C$ is single-crossing. 
\end{proof}

\subsection{Approximate Pattern Composition}\label{subsec:composition}

We define patterns induced by approximate patterns. This definition is similar to the patterns induced by patterns of~\cite{FMW20}. Recall the distance between a pattern and a vertex is defined in \Cref{def:patter-ver-dist}. 

\begin{definition}[Pattern Induced by an Approximate Distance]\label{def:patter-induced} Let  $\sigma$ be a sequence of $k$ vertices on the boundary of a face $f$ of a plane graph $G$. Let $\vec{p}$ be an approximate pattern (w.r.t. some sequence of vertices that may be different from $\sigma$). The pattern induced by  $\vec{p}$ w.r.t. $\sigma$ is a $(k-1)$-dimensional vector $\widehat{\vec{p}}$ where $\widehat{\vec{p}}[i] = \tilde{d}_G(\vec{p}, \sigma_{i+1}) - \tilde{d}_G(\vec{p}, \sigma_i)$ for every $i \in [k-1]$.  
\end{definition}

\begin{figure}[!htb]
	\begin{center}
		\includegraphics[width=.4\textwidth]{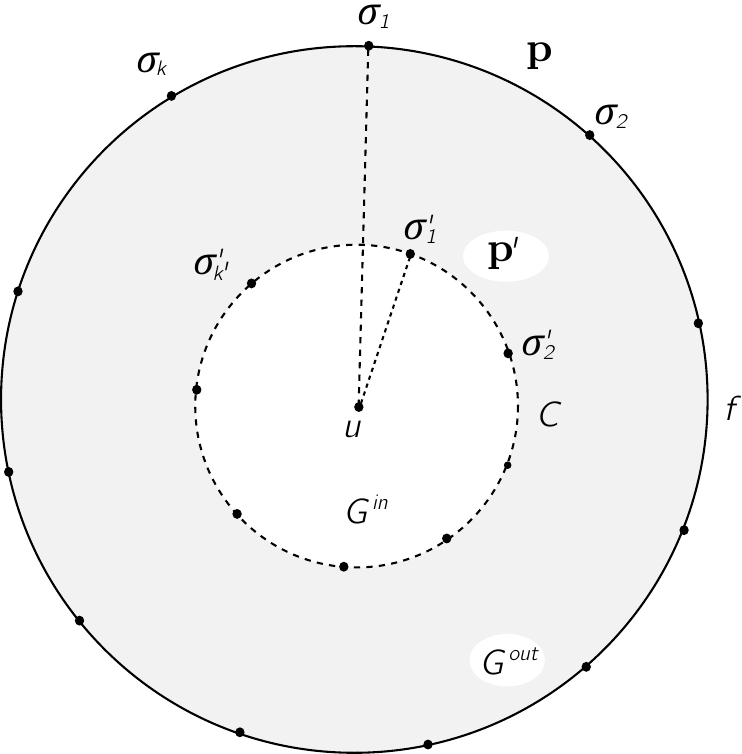}
	\end{center}
	\caption{Two patterns  $\vec{p}'$ and $\vec{p}$ of $u$ w.r.t. two sequences $\sigma'$ and $\sigma$ in $G^{in}$ and $G$, respectively. \Cref{lm:pattern-comp} implies that the pattern induced by $\vec{p}'$ in $G^{out}$ is a $k\delta$-approximation of  $\vec{p}$.}
	\label{fig:composition}
\end{figure}

In the following lemma, we show that, given a face $f$ of a plane graph, and a simple-separating cycle $C$ such that $f$ lies outside $C$, for any vertex $u$ inside $C$, we can construct a new pattern from $u$ w.r.t. some vertex sequence $\sigma$ in $f$ by \emph{composing} two patterns: a pattern from $u$ w.r.t. some vertex sequence $\sigma'$ on $C$, and the pattern w.r.t. $\sigma$ induced by a pattern w.r.t. $\sigma'$. The new pattern is may not be the same as the approximate pattern $\vec{p}$ of $u$ w.r.t. $\sigma$ in $G$ as defined in \Cref{def:appx-pattern}, but it is close to $\vec{p}$. See \Cref{fig:composition} for an illustration.

\begin{lemma}[Pattern Composition]\label{lm:pattern-comp}  Let  $\sigma$ be a sequence of $k$ vertices on the boundary of a face $f$ of a plane graph $G$, and $\sigma'$  be a sequence of $k'$  vertices on a single-crossing simple cycle $C$ such that $f\subseteq  \overline{\exter(C)}$. Let $G^{in}$ ($G^{out}$)  be the subgraph of $G$ induced by vertices inside (outside) or on $C$. Let $u$ be a vertex in $G^{in}$, $\vec{p}'$ be a pattern of $u$ in $G^{in}$ w.r.t. $\sigma'$, and  $\vec{p}$  be a pattern of $u$ in $G$ w.r.t. $\sigma$. Let $\widehat{\vec{p}}'$ be the pattern induced by $\vec{p}'$ in $G^{out}$.  Then, it holds that $\widehat{\vec{p}}'\approx_{\smallfont{k\delta}} \vec{p}$.
\end{lemma}
\begin{proof} By \Cref{def:patter-induced}, for any $i\in [k-1]$, we have:
	\begin{equation*}
		\begin{split}
			\widehat{\vec{p}}'[i] &= \tilde{d}_{G^{out}}(\vec{p}', \sigma_{i+1})   -  \tilde{d}_{G^{out}}(\vec{p}', \sigma_i) \\
			&= (\tilde{d}_{G}(u,\sigma_{i+1}) - \tilde{d}_{G^{in}}(u,\sigma'_1)) -  (\tilde{d}_{G}(u,\sigma_i) - \tilde{d}_{G^{in}}(u,\sigma'_1)) \qquad\mbox{(by  Item (2) in \Cref{lm:separating-cycle})}\\
			&= \tilde{d}_{G}(u,\sigma_{i+1}) - \tilde{d}_{G}(u,\sigma_i).
		\end{split}
	\end{equation*}
	By \Cref{eq:approx-dist}, we have:
	\begin{equation}
		\begin{split}
			\tilde{d}_{G}(u,\sigma_{i+1}) - \tilde{d}_{G}(u,\sigma_i) &\leq \appr{\delta}{d_{G}(u,\sigma_{i+1}) - d_{G}(u,\sigma_i) } + (i+1)\delta\\
			\tilde{d}_{G}(u,\sigma_{i+1}) - \tilde{d}_{G}(u,\sigma_i) &\geq \appr{\delta}{d_{G}(u,\sigma_{i+1}) - d_{G}(u,\sigma_i) } - (i+1) \delta~. 
		\end{split}
	\end{equation}	
Thus, $	\widehat{\vec{p}}'[i]\approx_{(i+1)\delta} \vec{p}[i]$, which  implies the lemma, since $i\leq k-1$. 
\end{proof}

\subsection{A Weaker Oracle: Proof of \Cref{thm:additiveOracleEasy}}\label{subsec:easyOracle}

In this section, we construct an oracle as claimed in \Cref{thm:additiveOracleEasy}. The construction is described in \Cref{subsubsec:const}, ignoring implementation issues. The analysis of query time and space is presented in \Cref{subsubsec:time-stretch} and \Cref{subsubsec:space-analysis}. Finally, in \Cref{subsubsec:preprocessing}, we discuss the implementation. 

\subsubsection{The Construction}\label{subsubsec:const}

First, we review the standard recursive decomposition using shortest path separators.  Let $T$ be a shortest path tree of $G$ rooted at a chosen vertex $r$. We assume w.l.o.g that $G$ is triangulated. A shortest path separator is a fundamental cycle $C$ of $T$ that comprises of two shortest paths rooted at $r$ and an edge $e$ connecting two other endpoints of the two paths. By \Cref{lm:single-crossing-sep}, $C$ is single-crossing.  It was known that, for any given non-negative weight function $w_V: V \rightarrow \mathbb{R}^+$, there is a shortest path separator such that the total weight of vertices strictly inside or outside $C$ is at most $\frac{2w_V(G)}{3}$ where $w_V(G) = \sum_{u\in V}w_V(u)$.

We use the shortest path separators to recursively separate $G$ into \emph{regions} as follows. Initially every vertex of $G$ is marked.  Starting from $G$, we construct a shortest path separator $C$ such that the total number of vertices strictly inside/outside $C$ is at most $\frac{2n}{3}$, obtaining two regions  sharing the same boundary $C$. We then distribute marked vertices on $C$ evenly to the two regions, so that each region gets at most $2n/3$ marked vertices. Some vertices on $C$ in one region get unmarked because they are marked vertices in the other region.
Next, pick any region $R$ that has at least $\lambda > 0$ marked vertices for some parameter $\lambda$ (defined in the algorithm below), we separate $R$ using the shortest path separator $C_R$ into two smaller regions $R_1,R_2$. The separator $C_R$ is chosen to balance the number of holes or the number of marked vertices inside child regions.  In particular, if $R$ has exactly 5 holes, we choose $C_R$ such that child regions $R_1,R_2$ each has at most $\lfloor \frac{2\cdot 5}{3} \rfloor + 1 = 4$ holes; this can be done by assigning weights to vertices on the hole appropriately. Otherwise,  we choose $C_R$ such that the number of marked vertices of each child region is at most  $2/3$ the number of marked vertices of $R$; again, we distribute marked vertices on $C_R$ evenly to both sides.   It follows from the construction that each region has at most 5 holes.

Let $\mathcal{T}$ be the recursion tree induced by the recursive decomposition of $G$. Each node of $\mathcal{T}$ is associated with a region resulting from the decomposition. It is well known that:

\begin{lemma}\label{lm:T-size} $\mathcal{T}$ has depth $O(\log(n))$ and $O(\frac{n}{\lambda})$ nodes that can be constructed in $O(n\log n)$ time. 
\end{lemma}

By construction, the internal vertices of a region are all marked. Unmarked vertices are on the boundaries of the holes.  We say that a region $A$ is an ancestor of a region $R$ if $A$ is associated with an ancestor node of $R$'s node in $\mathcal{T}$. Clearly, $R\subseteq A$.  Each region $R$, except $G$, has a special hole $h$ whose boundary is the  separating cycle $C_P$ of its parent region $P$. We call $h$ the \emph{parental hole} of $R$ (see \Cref{fig:recursion}(a)). (The parental hole of $R$ could be the infinite face of $R$ in the planar embedding inherited from $G$.) The boundary of $h$ is called the \emph{parental boundary} of $R$. In the following, we construct our distance oracle in four steps. \Cref{fig:orale-connst} and \Cref{fig:recursion} illustrates each step of the construction.

\begin{figure}[!htb]
	\begin{picture}(1000,450)
		\put(0,0){\includegraphics[width=1.0\textwidth]{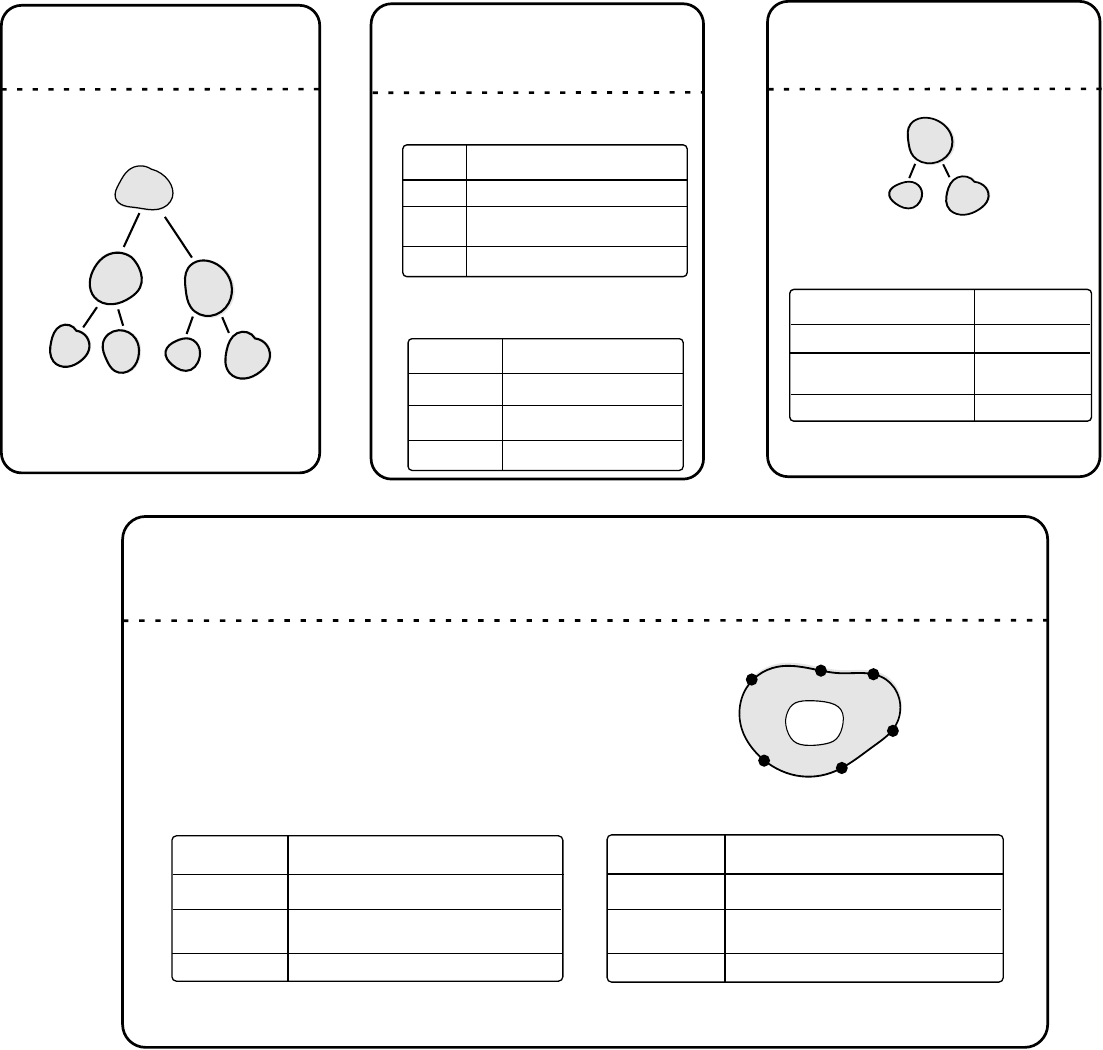}}
		\put(35,420){\huge\textbf{Step 1}}
		\put(45,390){$\mathcal{T}, \mathcal{LCA}_{\mathcal{T}}$}
		\put(23,265){ $|V(\mathcal{T})| = O(n/\lambda)$}
		\put(12,252){ $\lambda = \log(n)/\eps^{c},\quad c = 24$}
		\put(200,420){\huge\textbf{Step 2}}
		\put(180,392){$T_{21}$ stores $S_R,\vec{P}_R$}
		\put(175,375){\small Key}
		\put(230,375){\small Value}
		\put(175,364){ $\ldots$}
		\put(220,364){ $\ldots$}
		\put(175,347){ $u$}
		\put(200,348){\footnotesize $\appr{\delta}{d_{R^+}(u,\sigma_1)},\mbox{\footnotesize ID}(\vec{p}_u)$}
		\put(175,335){ $\ldots$}
		\put(220,335){ $\ldots$}
		
		\put(180,310){$T_{22}$ stores $\vec{D}^+_R$}
		\put(180,292){\small Key}
		\put(230,292){\small Value}
		\put(180,279){ $\ldots$}
		\put(220,279){ $\ldots$}
		\put(175,263){$s,\mbox{\footnotesize ID}(\vec{p})$}
		\put(230,263){$\mbox{\footnotesize ID}([s,\vec{p}]^\intercal)$}
		\put(180,251){ $\ldots$}
		\put(220,251){ $\ldots$}
			
		\put(370,420){\huge\textbf{Step 3}}
		\put(409,392){$R$}
		\put(370,352){$R_1$}
		\put(420,352){$R_2$}
		
		\put(370,329){$T_{3}$}
		\put(350,313){\small Key}
		\put(420,313){\small Value}
		\put(350,301){ $\ldots$}
		\put(420,301){ $\ldots$}
		\put(340,284){$\mbox{\footnotesize ID}(\vec{d}_1),\mbox{\footnotesize ID}(\vec{d}_2)$}
		\put(420,284){$\norm{\vec{d}_1,\vec{d}_2}$}
		\put(350,273){ $\ldots$}
		\put(420,273){ $\ldots$}
		\put(350,253){$\vec{d}_1\in \vec{D}^+_{R_1},\vec{d}_2\in \vec{D}^+_{R_2}$}
		
		\put(200,198){\huge\textbf{Step 4}}
		\put(80,158){$\vec{P}_A:$ pattern set of $A$}
		\put(80,142){$\vec{P}_R:$ pattern set of $R$}
		\put(80,127){$\widehat{\vec{P}}_{R,A}:$ patterns induced by patterns in $\vec{P}_R$}
		
		\put(385,158){$A$}
		\put(350,167){$\sigma^A$}
		\put(342,135){$R$}
		\put(362,142){\small $R^{out}$}
		
		\put(100,98){$T_{41a}$}
		\put(80,79){\small Key}
		\put(140,79){\small Value}
		\put(80,66){ $\ldots$}
		\put(145,66){ $\ldots$}
		\put(80,47){$\mbox{\footnotesize ID}(\vec{p})$}
		\put(140,47){$\tilde{d}_{R^{out}}(\vec{p}, \sigma^A_1),\mbox{\footnotesize ID}(\widehat{\vec{p}})$}
		\put(80,35){ $\ldots$}
		\put(145,35){ $\ldots$}
		\put(80,13){ $\vec{p}\in \vec{P}_R, \qquad \widehat{\vec{p}}\in \widehat{\vec{P}}_{R,A}$}
		
		\put(290,98){$T_{41b}$}
		\put(270,79){\small Key}
		\put(330,79){\small Value}
		\put(270,66){ $\ldots$}
		\put(335,66){ $\ldots$}
		\put(270,47){$\mbox{\footnotesize ID}(\widehat{\vec{p}})$}
		\put(330,47){$\mbox{\footnotesize ID}(\vec{p}')$}
		\put(270,35){ $\ldots$}
		\put(335,35){ $\ldots$}
		\put(250,13){ $\widehat{\vec{p}}\in \widehat{\vec{P}}_{R,A},\vec{p}'\in \vec{P}_A, \lVert\widehat{\vec{p}} - \vec{p}'\rVert_{\infty}$ minimum}

	\end{picture}
	\caption{Tables and data structures used in each step of the oracle construction.}
	\label{fig:orale-connst}
\end{figure}
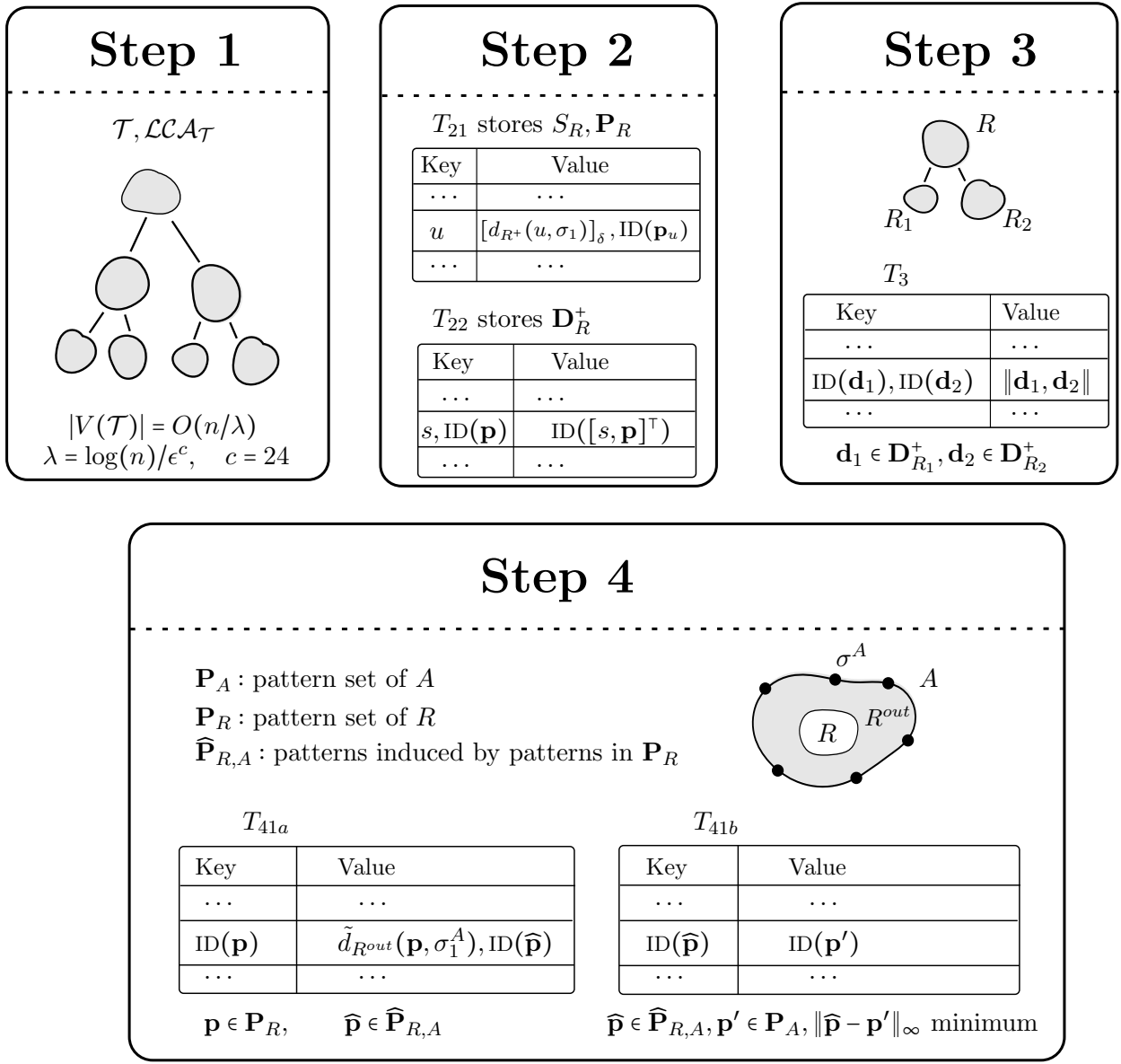

\pagebreak

\newgeometry{left=0.5in,bottom=0.8in,top=0.8in,right=0.5in}

\begin{tcolorbox}
\begingroup
\fontsize{10pt}{12pt}\selectfont
\textsc{Distance Oracle Construction:}
\paragraph{Step 1.} Construct a recursive decomposition $\mathcal{T}$ of $G$ with $\lambda = \frac{\log n}{\eps^{c}}$ for $c = 24$, and a lowest common ancestor data structure $\lca_{\mathcal{T}}$ for $\mathcal{T}$.  For each shortest path separator $C$ in $\mathcal{T}$, we designate $O(1/\eps)$ portals such that the distance between any consecutive portals is $\eps D$;  if there is an edge of length more than $\eps D$, we will subdivide the edge using portals. For each region $R$ associated with a node of $\mathcal{T}$, let $M(R)$ be the set of marked vertices of $R$.

\paragraph{Step 2.}  Let $\delta = \eps^3 D$. For each region $R$ associated with a non-root node of $\mathcal{T}$, let $R^+$ be obtained from $R$ by \emph{filling all the holes of $R$, except the parental hole of $R$}. That is, for each non-parental hole $h$, we add the edges and vertices of $G$ inside $h$ to $R$.   We form a sequence $\sigma$ of the portals  on the parental boundary of $R$ by ordering the portals clockwise order.  (See \Cref{fig:recursion}(b).) For each vertex $u\in M(R)$, let $\vec{p}_u$  be the $\delta$-approximate pattern  of $u$ w.r.t. $\sigma$ \emph{in graph $R^+$}. Let $\vec{P}_{R} = \{\vec{p}_u: u\in M(R)\}$ and $S_R = \{ \appr{\delta}{d_{R^+}(u,\sigma_1): u\in M(R)}\}$.  
Let $\vec{D}^+_R = S_R\times \vec{P}_R$ be a set $k$-dimensional vectors whose first coordinates are scalars in $S_R$; here $k = |\sigma|$.
We  assign unique IDs to vectors in $\vec{P}_R$ and $\vec{D}^+_R$. At node $R$, we store $S$ and $ \vec{P}_R$ in a table $T_{21}$ and  $\vec{D}^+_R$ in a table $T_{22}$.  Given $u$, we can query $\appr{\delta}{d_{R^+}(u,\sigma_1)}$ and the ID of $\vec{p}_u$  in $T_{21}$ in $O(1)$ time. Given a scalar $s\in S_R$ and an ID of a vector $\vec{p} \in \vec{P}_{R}$, we can query the ID of the corresponding vector $[s,\vec{p}]^{\intercal}\in \vec{D}^+_R$ in  $T_{22}$ in $O(1)$ time.

\paragraph{Step 3.} For each non-leaf node $R$ with two child regions $R_1$ and $R_2$. Observe that $R_1^+\cup R_2^+ = G$ by definition, and that $R_1^+\cap R_2^+ = C_R$ where $C_R$ is the shared parental boundary of $R_1$ and $R_2$. (See \Cref{fig:recursion}(a).)  Let $\sigma$ be the sequence of portal vertices on $C_R$.  For each pair of distance encodings $\vec{d}_1 \in \vec{D}^+_{R_1}$ and $\vec{d}_2 \in \vec{D}^+_{R_2}$, we compute $\norm{\vec{d}_1, \vec{d}_2}$  using~\Cref{def:encoding-dist}, and store it in a lookup table $T_{3}$, indexed by the IDs of  $\vec{d}_1$ and $\vec{d}_2$.  Thus, given the IDs of $\vec{d}_1$ and $\vec{d}_2$, we can look up the distance $\norm{\vec{d}_1, \vec{d}_2}$ in $T_{3}$ in $O(1)$ time.	

\paragraph{Step 4.}   For each \emph{leaf region} $R$, we store the following information at the node of $R$.
\begin{itemize}
	\item[(4.1)] For each ancestor region $A$ of $R$, let $R^{out}$ be the graph induced by the edge set $(E(A^+)\setminus E(R^+))\cup \partial R^+$. Let $\sigma^A$  be the sequence of portals of the parental boundary of $A$; see \Cref{fig:recursion}(c). Let $\widehat{\vec{P}}_{R,A}$ be the set of patterns w.r.t. $\sigma^A$ induced by $\delta$-approximate patterns in $\vec{P}_R$ in graph $R^{out}$.
	\begin{itemize}
		\item[(a)] We store in a table $T_{41a}$ for each pattern $\vec{p} \in \vec{P}_R$ the distance $\tilde{d}_{R^{out}}(\vec{p}, \sigma^A_1)$, and the ID of the pattern $\widehat{\vec{p}}\in \widehat{\vec{P}}_{R,A}$ induced by $\vec{p}$. Given the ID of $\vec{p}$, we can access the corresponding entry of $T_{41a}$ in $O(1)$ time.
		\item[(b)] Recall that $\vec{P}_A$ is the set of $\delta$-approximate patterns of A computed in Step 2. For each induced pattern $\widehat{\vec{p}}\in \widehat{\vec{P}}_{R,A}$, we store  in a table $T_{41b}$ the ID of the pattern $\vec{p}'\in \vec{P}_A$ closest to $\widehat{\vec{p}}$. That is, $\vec{p}'$ minimize  $\lVert |\widehat{\vec{p}} - \vec{p}'\rVert_{\infty}$ over all patterns in $ \vec{P}_A$.	We can access entries of $T_{41b}$ in $O(1)$ time given the ID of $\widehat{\vec{p}}$.  
	\end{itemize}
	\item[(4.2)] We construct a graph $R'$, called a \emph{contracted filled graph}, from $R$ as follows. (See \Cref{fig:recursion}(d).) For each hole $h$ of $R$, let $P_h$ be the set of portals on the boundary of $h$ constructed in Step 1. Let $G^h$ be the graph induced by edges of $G$ inside and on the boundary of $h$. Let $K^h$ be the distance preserving minor of $P_h$ in $G^h$. 
	Next, let $\overline{R}$ be obtained from $R$ by contracting each unmarked vertices on the boundaries of the holes of $R$ to the nearest portal. The weight of new edges between two portals on the boundary of $\overline{R}$ is their distance in $R$. For every set of parallel edges whose one endpoint is non-portal, we keep  the minimum-weight edge. Let:
	\begin{equation}\label{eq:Rprime}
		R' = \overline{R}\bigcup (\cup_{h \mbox{ is a hole of }R} K^h).
	\end{equation}
	Observe that $R'$ is planar. We apply \Cref{thm:LongPettie} to construct an exact distance oracle $\mathcal{E}_R$ for $R'$ with:
	\begin{itemize}
		\item[(a)]  $|V(R')|^{1+o(1)}$ space  and  $\log^{2+o(1)}(|V(R')|)$ query time or
		\item[(b)] $|V(R')|\log^{2+o(1)}(|V(R')|)$ space and $|V(R')|^{o(1)}$ query time. 
	\end{itemize}
\end{itemize}
\endgroup
\end{tcolorbox}
\restoregeometry

\begin{figure}[!htb]
	\includegraphics[width=1.0\textwidth]{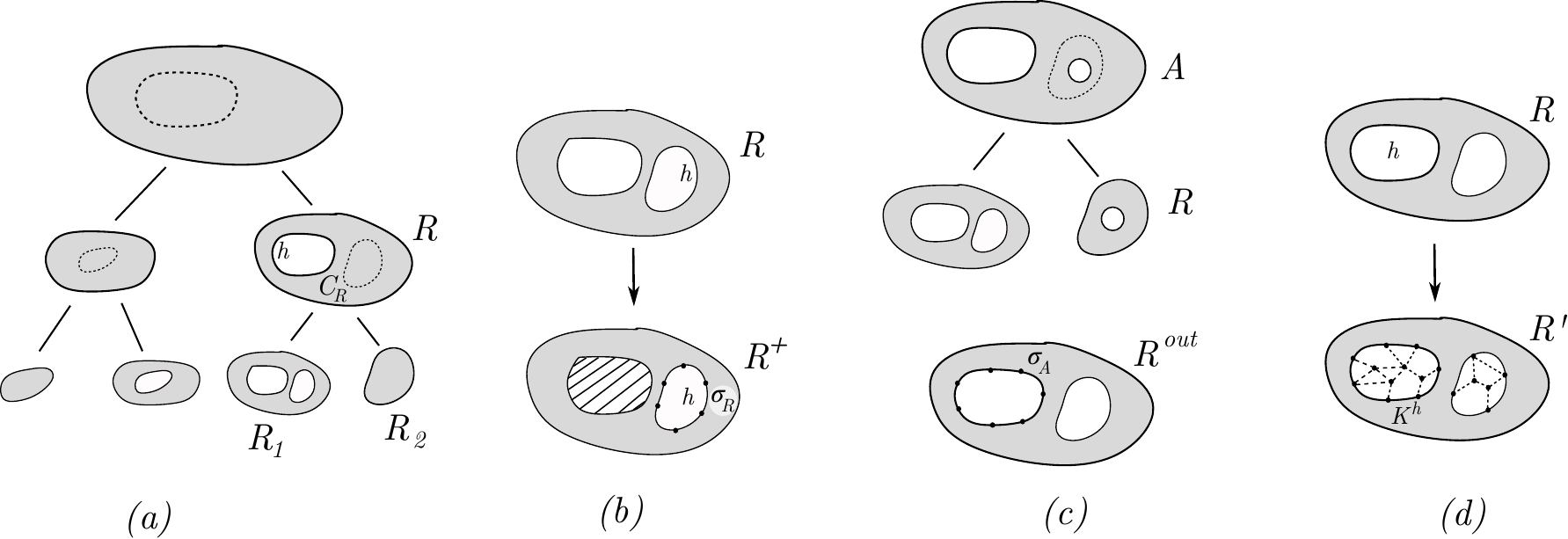}
	\caption{(a) A recursive decomposition tree $\mathcal{T}$; $h$ is the parental hole of $R$ and the boundary of $h$ is the parental boundary of $R$; separator $C_R$ separates $R$ into two children $R_1$ and $R_2$. (b) $R$ has two holes, where $h$ is the parental hole; $R^+$ is obtained from $R$ by filling up the non-parental hole; $\sigma^R$ is the sequence of vertices on the boundary of $h$. (c) A region $R$ and its parent $A$ (top) and the subgraph $R^{out}$ induced by the edge set $(E(A^+)\setminus E(R^+))\cup \partial R^+$ (bottom). (d) The contracted filled graph $R'$ of $R$ is obtained by first contracting some edges on the boundary of $R$ to obtain $\overline{R}$ and then filling each hole $h$ in $\overline{R}$ by a graph $K^h$.}
	\label{fig:recursion}
\end{figure}

Here we briefly describe the idea of each step in the construction. Step 1 constructs the recursive decomposition $\mathcal{T}$ and a basic data structure to navigate $\mathcal{T}$. Step 2 stores the set of approximate patterns of each region $\vec{P}_R$ and an \emph{extended} set of approximate distance decodings $\vec{D}_R^+$. The set of approximate distance decodings of $R$ is  $\vec{D}_R = \{\vec{d}_u: u\in M(R)\}$, which is a subset of $\vec{D}_R^+$. The reason we store the extended set is that we could not query the approximate distance decoding $\vec{d}_u$ of $u$ from the information stored at the leaf region containing $u$ due to the pattern composition step. The pattern composition only gives an approximate distance decoding $\tilde{d}_u$ that is close to $\vec{d}_u$ (in $\ell_{\infty}$ norm) and $\tilde{d}_u \in \vec{D}_R^+$. Step 3 precomputes  the distance of every pair of approximate distance decodings in a table so that we can look up the distance in $O(1)$ time during the query stage. Step 4(1) implements the pattern composition discussed in \Cref{subsec:composition}, and Step 4(2) constructs an exact distance oracle for each leaf region.

\begin{wrapfigure}{r}{0.4\textwidth}
	\vspace{-40pt}
	\begin{center}
		\includegraphics[width=0.38\textwidth]{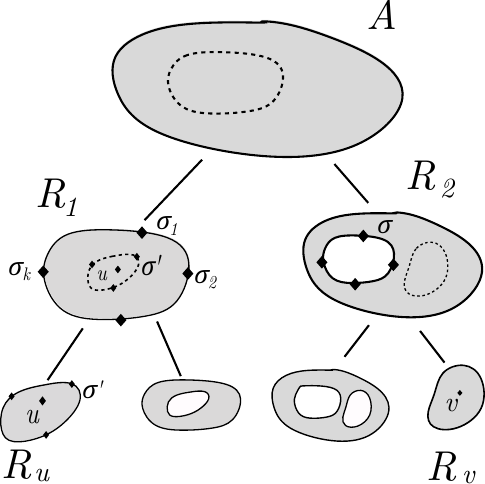}
	\end{center}
	\vspace{-20pt}
\end{wrapfigure}

\paragraph{Answering a query.~} Let $u$ and $v$ be two given vertices in the query. First, we find two leaf regions $R_u$ and $R_v$ containing $u$ and $v$ such that $u\in M(R_u)$ and $v\in M(R_v)$. If $R_u = R_v = R$, we query the distance between $u$ and $v$ in $\mathcal{E}_{R}$, denoted by $\mathcal{E}_{R}(u,v)$. We then return $\mathcal{E}_{R}(u,v) + 2\eps D$ as the approximate distance. (We add $+2\eps D$ to  $\mathcal{E}_{R}(u,v)$ because of the contraction in the construction of $\overline{R}$.)

We now assume that $R_u\not= R_v$.  Let $A  =\lca_{\mathcal{T}}(R_u,R_v)$ be the lowest common ancestor region of $R_u$ and $R_v$. Let $R_1$ and $R_2$ be two child regions of $A$ where $R_1$ ($R_2$) is an ancestor of $R_u$ ($R_v$). Let $\sigma$  be the  sequence of vertices on the parental boundary of $R_1$ and $R_2$. By construction, $R_1$ and $R_2$ share the same parental boundary. Next, we use the following lemma, whose proof is deferred to the end of this section.

\begin{lemma}\label{lm:query-approximiate} Given $u\in R_u$, we can query the ID of  a vector $\tilde{\vec{d}}_u \in \vec{D}^+_{R_1}$  in $O(1)$ time such that  $\tilde{\vec{d}}_u\approx_{\smallfont{0,2k\delta}}\vec{d}_u$ where  $\vec{d}_u$ is the approximate distance encoding of $u$ w.r.t. $\sigma$ in $R_1^{+}$.
\end{lemma}

We apply \Cref{lm:query-approximiate} to query the IDs of approximate distance encodings $\tilde{\vec{d}}_u \in \vec{D}^+_{R_1}$ and  $\tilde{\vec{d}}_v\in \vec{D}^+_{R_2}$  of $u$ and $v$, respectively.   Given the two approximate distance encoding IDs, we query table $T_3$ of $A$ in $O(1)$ time to obtain $\norm{\tilde{\vec{d}}_u, \tilde{\vec{d}}_v}$. Finally, we return:
\begin{tcolorbox}
	\begin{equation}\label{eq:returned-distance}
		\norm{\tilde{\vec{d}}_u, \tilde{\vec{d}}_v} + c_0\eps D \qquad \mbox{for some constant $c_0$ chosen later}
	\end{equation}
\end{tcolorbox}
 as an approximate distance between $u$ and $v$ of our oracle.

\begin{proof}[Proof of \Cref{lm:query-approximiate}]  Let $R_u^{out}$ be   the graph induced by the edge set $(E(R_1^+)\setminus E(R_u^+))\cup \partial R_u^+$. 
	($R_u^{out}$ is exactly the graph in Step 4(1) obtained from the ancestor region $R_1$ and  the leaf region $R_u$.)   
    \InsertBoxR{2}{\begin{minipage}{0.45\linewidth}\centering
   		\vspace{10pt}
   		\includegraphics{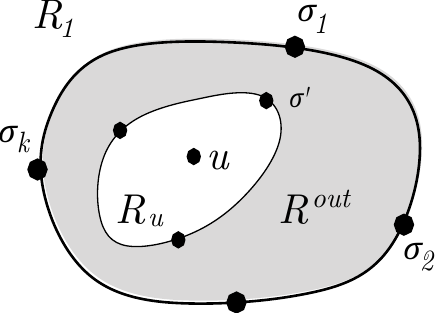}
   		\vspace{-20pt}
   	\end{minipage}%
   }[5]

	Recall each vector $\tilde{\vec{d}}_u \in \vec{D}^+_{R_1}$ can be written as $\tilde{\vec{d}}_u  = \left[\appr{\delta}{d_{R_1^+}(u,\sigma_1)}, \vec{p}\right]^\intercal$ for some approximate pattern $\vec{p} \in \vec{P}_{R_1}$. (Pattern $\vec{p}$ might be different from $\vec{p}_u$, the approximate pattern of $u$ in $\vec{P}_{R_1}$.) Furthermore, we can query the ID of $\tilde{\vec{d}}_u$ from table $T_2$ stored at $R_1$ constructed in Step 2. To do so, we need to know the value of $\appr{\delta}{d_{R_1^+}(u,\sigma_1)}$ and the ID of $\vec{p}$.

	Recall by \Cref{remrk:dusigma1}, $\appr{\delta}{d_{R_1^+}(u,\sigma_1)} = \tilde{d}_{R_1^+}(u,\sigma_1)$. We query $\tilde{d}_{R_1^+}(u,\sigma_1)$ as follows. Let $\sigma'$ be the sequence of portals in the parental hole of $R_u$.   We use $u$ to query table $T_{2}$ of $R_u$ (constructed in Step 2)  to get $\appr{\delta}{d_{R_u^+}(u,\sigma'_1)}$ and the approximate pattern of $u$ in graph $R^+_u$, denoted by $\vec{p}'_{u}$. (We reserve $\vec{p}_u$ notation for the approximate pattern of $u$ in $R^+_1$.)   Next, we use the ID of $\vec{p}'_{u}$ to query $\tilde{d}_{R^{out}}(\vec{p}'_u,\sigma_1)$ in table $T_{41a}$ (stored at $R_u$) constructed in Step 4(1a). The total query time is $O(1)$. Observe that:

	\begin{equation*}
		\begin{split}
			&\appr{\delta}{d_{R_u^+}(u,\sigma'_1)} + \tilde{d}_{R^{out}}(\vec{p}_u',\sigma_1 )\\
			&=  \tilde{d}_{R_u^+}(u,\sigma_1') + \tilde{d}_{R^{out}}(\vec{p}_u',\sigma_1) \qquad \mbox{(by \Cref{remrk:dusigma1})}\\
			&= \tilde{d}_{R_1^+}(u,\sigma_1) \qquad \mbox{(by Item (2) of \Cref{lm:separating-cycle})}
		\end{split}
	\end{equation*}
	 which will give us the value of $\appr{\delta}{d_{R_1^+}(u,\sigma_1)}$.

	Next, we query the ID of $\vec{p}$ as follows. First, we query the ID of $\vec{p}'_{u}$ stored in table $T_2$ of $R_u$ as above. Then, we use the ID of $\vec{p}'_{u}$ to query the ID of an approximate pattern $\vec{p}$ in table $T_{41b}$ for ancestor $R_1$ of $R_u$ constructed in Step 4(1b) in $O(1)$ time.
	
	Now using $\appr{\delta}{d_{R_1^+}(u,\sigma_1)}$ and the ID of $\vec{p}$, we query the ID of $\tilde{\vec{d}}_u \in \vec{D}^+_{R_1}$ from table $T_2$ stored at $R_1$ constructed in Step 2. Note that $\tilde{\vec{d}}_u[1] = \appr{\delta}{d_{R_1^+}(u,\sigma_1)} = \tilde{d}_{R_1^+}(u,\sigma_1)$ and $\tilde{\vec{d}}_u[2:k] = \vec{p}$.

	To complete the proof of the lemma, we show that $\tilde{\vec{d}}_u\approx_{\smallfont{0,2k\delta}}\vec{d}_u$. Observe that  $\tilde{\vec{d}}_u[1] = \appr{\delta}{d_{R_1^+}(u,\sigma_1)}=  \vec{d}_u[1]$. Thus, by \Cref{eq:encoding-close}, it remains to show that $\tilde{\vec{d}}_u[2:k] \approx_{2k\delta} \vec{d}_u[2:k]$ which is equivalent to showing that $ \vec{p}\approx_{2k\delta}\vec{p}_u$.

	By construction in Step 4(1b), $\vec{p}$ is the closest pattern of the pattern, say $\widehat{\vec{p}}_u'  \in \widehat{\vec{P}}_{R_u,R_1}$, induced by  $\vec{p}_u'$.   By \Cref{lm:pattern-comp}, $\widehat{\vec{p}}_u'\approx_{\smallfont{k\delta}} \vec{p}_u$; that is, $\lVert\widehat{\vec{p}}_u', \vec{p}_u\rVert_\infty \leq k\delta$.  Since $\vec{p}$  is closest to $\widehat{\vec{p}}_u'$ in $\ell_{\infty}$ norm, we have:
	\begin{equation}
		\lVert \vec{p}, \widehat{\vec{p}}_u' \rVert_{\infty } \leq \lVert\widehat{\vec{p}}_u', \vec{p}_u\rVert_\infty \leq k\delta
	\end{equation}
	Since $\vec{p} \approx_{k\delta}  \widehat{\vec{p}}_u' $ and $\widehat{\vec{p}}_u'\approx_{\smallfont{k\delta}} \vec{p}_u$, it follows from \Cref{obs:apprx} that $ \vec{p}\approx_{2k\delta}\vec{p}_u$, as desired.
\end{proof}

\subsubsection{Query time and stretch analysis}\label{subsubsec:time-stretch}

\paragraph{Query time analysis.~} Let $u$ and $v$ be two querying vertices. If $R_u = R_v = R$, the algorithm queries  $\mathcal{E}_{R}$ in time:

\begin{equation}\label{eq:time-Exact}
	\begin{cases} \log^{2+o(1)}(|V(R_u)|) = \log^{2+o(1)}(1/\eps) + \log^{2 + o(1)}\log(n)  \qquad &\mbox{ in Regime 4(2a)}\\
		|V(R_u)|^{o(1)} = \log^{o(1)}(n)(1/\eps)^{o(1)}  \qquad &\mbox{ in Regime 4(2b)}	
	\end{cases}
\end{equation}

We now consider the case $R_u\not= R_v$. The time to compute  $A  =\lca_{\mathcal{T}}(R_u,R_v)$ is  $O(1)$. Let $R_1$ and $R_2$ be two child regions of $A$ where $R_1$ ($R_2$) is an ancestor of $R_u$ ($R_v$).  By \Cref{lm:query-approximiate}, the time to query the IDs of approximate distance encodings   $\tilde{\vec{d}}_u \in \vec{D}^+_{R_1}$ and    $\tilde{\vec{d}}_v \in \vec{D}^+_{R_2}$ of $u$ and $v$, respectively, is $O(1)$. Next, we query table $T_3$ of $A$ in $O(1)$ time to obtain $\norm{\tilde{\vec{d}}_u,\tilde{\vec{d}}_v}$. Finally, we return the approximate distance in \Cref{eq:returned-distance} in $O(1)$ time. 

In summary, the total query time is dominated by the query time in \Cref{eq:time-Exact}, which implies the query time claimed in \Cref{thm:additiveOracleEasy}.

\paragraph{Stretch analysis and choosing $c_0$ in \Cref{eq:returned-distance}.~} If $u$ and $v$ are in the same leaf region $R$, then the exact distance query in $\mathcal{E}_{R}$ will return the exact distance between $u$ to $v$ in $R'$. We show in the following that this distance approximates the true distance in $G$. 

\begin{lemma}\label{clm:dRprime}$d_G(u,v) - 2\eps D\leq d_{R'}(u,v) \leq d_G(u,v)+ 10\eps D$.
\end{lemma}
\begin{proof}
	We reuse the notation in Step 4(2) here. As the distance of each unmarked vertex to the nearest portal is at most $\eps D$, by the triangle inequality,  for any two marked vertices $u$ and $v$, it holds that:
	\begin{equation}\label{eq:dRoverline}
		d_{R}(u,v) - 2\eps D \leq d_{\overline{R}}(u,v) \leq d_{R}(u,v)
	\end{equation}
	Let $\overline{G}$ be obtained from $\bar{R}$ by filling each hole exactly. That is:
	\begin{equation*}
		\overline{G} = \overline{R}\bigcup (\cup_{h \mbox{ is a hole of }R} G^h).
	\end{equation*}
	It follows from \Cref{eq:dRoverline} that:
	\begin{equation}\label{eq:dGoverline}
		d_{G}(u,v) - 2\eps D \leq d_{\overline{G}}(u,v) \leq d_{G}(u,v)
	\end{equation}
	Since $R'$ is obtained from $\overline{R}$ by approximately filling 5 holes through portals, and since the distance between any two consecutive portals is $\eps D$, by the triangle inequality,  $d_{\overline{G}}(u,v)\leq d_{R'}(u,v) \leq d_{\overline{G}}(u,v) + 10\eps D$. Combined with \Cref{eq:dGoverline}, we have $	d_{G}(u,v) - 2\eps D  \leq d_{R'}(u,v) \leq d_{G}(u,v) + 10\eps D$ as claimed.
\end{proof}

The distance that the oracle returns is $\mathcal{E}_R(u,v) + 2\eps D  = d_{R'}(u,v)+2\eps D$, which is at least $d_G(u,v)$ and at most $d_G(u,v) + 12\eps D$ by \Cref{clm:dRprime}.

It remains to consider the case $R_u\not=R_v$. Observe that for every region $R$ in $\mathcal{T}$, the sequence $\sigma$ of vertices forming from the  portals of the parental boundary of $R$ is a $\eps D$-cover of the boundary. Thus, $\tau = \eps D$. Since every shortest path has length at most $D$,  there are at most $k = O(D/\tau) = O(1/\eps)$ vertices in  $\sigma$. 

In \Cref{eq:returned-distance},  the oracle computes and returns $ \norm{\tilde{\vec{d}}_u,\tilde{\vec{d}}_v} + c_0\eps D$. The following lemma, whose proof is deferred to the end of this section, will help us bound the stretch.

\begin{lemma}\label{lm:dhat-uv} $\norm{\tilde{\vec{d}}_u,\tilde{\vec{d}}_v} \approx_{\smallfont{(4k^2-2k)\delta + 2\tau}}d_G(u,v)$.  In particular, when $ k = O(1/\eps), \delta = \eps^3 D$ and $\tau = \eps D$, then $\norm{\tilde{\vec{d}}_u,\tilde{\vec{d}}_v} \approx_{\smallfont{c_0 \eps D}}d_G(u,v)$ for some constant $c_0$.
\end{lemma}

We choose $c_0$ in \Cref{eq:returned-distance} to be the value in \Cref{lm:dhat-uv}. It follows that:
\begin{equation*}
	d_G(u,v)\leq \norm{\tilde{\vec{d}}_u,\tilde{\vec{d}}_v} + c_0\eps D \leq d_G(u,v) + 2c_0\eps D~.
\end{equation*}
That is, the additive stretch of our oracle is $+O(\eps D)$. We can get back additive stretch $+\eps D$ by scaling $\eps$.

\begin{proof}[Proof of \Cref{lm:dhat-uv}] 
	By construction, $R_1^+ \cup R_2^+ = G$. We denote $G^{in} = R_1^+ $ and $G^{out} = R_2^+$.    The cycle separating $G^{in}$ and $G^{out}$ is single-crossing by \Cref{lm:single-crossing-sep}. Recall that $\sigma$ is the sequence of at most $k$ portals on the (shared) parental boundary of $R_1$ and $R_2$.  Let $\vec{p}_u$ ($\vec{d}_u$) and $\vec{p}_v$ ($\vec{d}_v$) be the $\delta$-approximate patterns (distance encodings) of $u$ and $v$ w.r.t. $\sigma$ in   $G^{in}$ and  $G^{out}$, respectively. By \Cref{lm:query-approximiate}, we have:
	\begin{equation*}
			\tilde{\vec{d}}_u \approx_{\smallfont{0},\smallfont{2k\delta}} \vec{d}_u  \qquad \mbox{and} \qquad 
			\tilde{\vec{d}}_v \approx_{\smallfont{0},\smallfont{2k\delta}}  = \vec{d}_v
	\end{equation*}
	By applying \Cref{lm:dist-approximation-enco} with $\delta_1 =0$ and $\delta_2 = 2k\delta$, $\norm{\tilde{\vec{d}}_u,\tilde{\vec{d}}_v} \approx_{\smallfont{4(k-1)k\delta}} \tilde{d}_G(u,v)$. By Item (1) in  \Cref{lm:separating-cycle}, $\tilde{d}_G(u,v)\approx_{\smallfont{2k\delta + 2\tau}}d_G(u,v)$. It follows from \Cref{obs:apprx} that $\norm{\tilde{\vec{d}}_u,\tilde{\vec{d}}_v}\approx_{\smallfont{(4k^2-2k)\delta + 2\tau}}d_G(u,v)$ as desired.
\end{proof}

\subsubsection{Space analysis}\label{subsubsec:space-analysis}

 First, we bound the number of $\delta$-approximate patterns w.r.t. sequence $\sigma_H $ of at most $k = O(1/\eps)$ vertices in a graph $H$ arising during the  construction of the oracle; $H$ could be $R^+$ in Step 2, or $R^{out}$ in Step 4(1). Recall that we set $\delta = \eps^3 D$ in Step 2.  Since the distance between $\sigma_{i+1}$ and $\sigma_i$ is at most $\eps D$ by construction for every $i\in [k-1]$, it follows from the triangle inequality that $-\eps D \leq d_G(u,\sigma_{i+1}) - d_G(u,\sigma_i) \leq \eps D$. That is, the constant $g$ in \Cref{lm:pattern-bound} satisfies $g \leq \frac{\eps D}{\delta} \leq 1/\eps^2$.  Let $\vec{P}_H$ be the set of all approximate patterns of $H$. By  \Cref{lm:pattern-bound}, we have:
\begin{equation}\label{eq:pattern-region}
	|\vec{P}_H| = O((kg)^3) = O(1/\eps^9)~.
\end{equation}

In the next four lemmas, we bound the space of each step of the oracle construction. (See \Cref{fig:orale-connst} for an illustration.) Recall that $c = 24$ specified in Step 1.

\begin{lemma}\label{lm:step1-space} The total space of Step 1 is $O(\frac{n\eps^{c-1}}{\log n})$.
\end{lemma}
\begin{proof}
	Since $\lambda = \frac{\log(n)}{\eps^c}$, $\mathcal{T}$ has $O(\frac{n\eps^{c}}{\log n})$ nodes by \Cref{lm:T-size}. Thus, the total space to store all portals in Step 1 is $O(\frac{n\eps^c}{\log n}\cdot (1/\eps)) = O(\frac{n\eps^{c-1}}{\log n})$. The LCA data structure $\lca_{\mathcal{T}}$ (see, e.g.,~\cite{BF00}) has space  $O(|V(\mathcal{T})|) = O(\frac{n\eps^c}{\log n})$, which implies the lemma.
\end{proof}

\begin{lemma}\label{lm:step2-space} The total space of Step 2 is $O(\frac{n\eps^{c-13}}{\log n})$.
\end{lemma}
\begin{proof}
	Observe that there are  $\frac{2D}{\delta} = O(1/\eps^3)$ different values of $\appr{\delta}{d_{R_1^+}(u,\sigma_1)}$ considered in Step 2. This is because the subgraph associated every node of $\mathcal{T}$ has a diameter at most $2D$ due to the special structure of regions separated by shortest path separators: the shortest path trees of these subgraphs are subtrees of $T$, the shortest path tree of $G$ rooted at $r$.  Since the total number of patterns for each region is $O(1/\eps^9)$ by \Cref{eq:pattern-region}, the total number of approximate distance encodings is $O(1/\eps^{12})$. Since each distance encoding has size $O(1/\eps)$, the space of table $T_{12}$ and $T_{22}$ in Step 2 is $O(1/\eps^{13})$. It follows that the total space in Step 2 is $O(\frac{n\eps^{c}}{\log n} (1/\eps^{13})) = O(\frac{n\eps^{c-13}}{\log n})$.
\end{proof}

\begin{lemma}\label{lm:step3-space} The total space of Step 3 is $ O(\frac{n\eps^{c-24}}{\log n})$.
\end{lemma}
\begin{proof}
	The analysis of Step 2 in \Cref{lm:step2-space} shows that the number of approximate distance encodings of each region is $O(1/\eps^{12})$. Thus, the space of table $T_{3}$ is $O((1/\eps^{12})^2) = O(1/\eps^{24})$. It follows that the total space of Step 3 is $O(\frac{n\eps^{c}}{\log n} (1/\eps^{24})) = O(\frac{n\eps^{c-24}}{\log n})$.
\end{proof}

In Step 4, we have two different regimes, namely Regime 4(2a) and Regime 4(2b), for the choice of the exact distance oracle $\mathcal{E}_R$ in Step 4(2). Our analysis considers each regime separately.

\begin{lemma}\label{lm:step4-space} The total space of Step 4 is $n(\log(n)1/\eps)^{o(1)}$ for Regime 4(2a) and $n(\log^{2+o(1)}(\log n) + \log^{2+o(1)}(1/\eps))$ for Regime 4(2b).
\end{lemma}
\begin{proof}
We consider each substep separately. We assume that $1/\eps = n^{O(1)}$; otherwise, applying the exact oracle in \Cref{thm:LongPettie} gives the desired bounds in \Cref{thm:main}. Thus, a value of size $\poly(1/\eps)$ costs $O(1)$ words of space to store.

\paragraph{Step 4(1)}  We observe by \Cref{eq:pattern-region} that $|\vec{P}_{R}| = O(1/\eps^9)$ and  $|\widehat{\vec{P}}_{R,A}| = O(1/\eps^9)$. Thus the ID of each pattern in $\vec{P}_{R}$ and $\widehat{\vec{P}}_{R,A}$ only costs $O(1)$ words. Storing $\widehat{\vec{P}}_{R,A}$ requires $O(1/\eps^{10})$ words of space since each pattern has size $O(1/\eps)$. In Step 4(1a), the  total space of $T_{41a}$ is $O(|\vec{P}_{R}|) = O(1/\eps^9)$. In Step 4(1b), the total space of $T_{41b}$ is also $O(|\widehat{\vec{P}}_{R,A}|) = O(1/\eps^9)$. Since each leaf region $R$ has $O(\log n)$ ancestors by \Cref{lm:T-size}, the total space of Step 4(1) is $O(\frac{n\eps^{c}}{\log n} (1/\eps^9) \log(n)) = O(n\eps^{c-9}) = O(n)$.

\paragraph{Step 4(2)}  Observe that $|M(R)| = \Theta(\lambda) = \Theta(\log(n)/\eps^{c})$. Since $R$ has at most $5$ holes and each hole we add $O(1/\eps^{4})$ in the construction of $R'$, it follows that:
\begin{equation}\label{eq:Rprime-size}
	|V(R')| = O(\log(n)/\eps^{c} + 1/\eps^4) = O(\log(n)/\eps^{c})~.
\end{equation}

There are two different regimes for the choice of $\mathcal{E}_R$:
\begin{itemize}
	\item[(a)] the space of $\mathcal{E}_R$ is $\log^{1+o(1)}(n)/\eps^{c + o(1)}$. Then the total space of Step 4(2) is: 
	\begin{equation}\label{eq:space-R1}
		O(\frac{n \eps^{c}}{\log(n)})\cdot \frac{ \log^{1+o(1)}(n)}{\eps^{c + o(1)}} = n(\log(n)1/\eps)^{o(1)}.
	\end{equation}
	\item[(b)]   the space of $\mathcal{E}_R$ is $\frac{\log(n)}{\eps^c}(\log^{2+o(1)}(\log(n)) + \log^{2+o(1)}(1/\eps))$. Then the total space of Step 4(2) is: 
	\begin{equation}\label{eq:space-R2}
		\begin{split}
			&O(\frac{n \eps^{c}}{\log(n)})\cdot \frac{\log(n)}{\eps^c}\left(\log^{2+o(1)}(\log(n)) + \log^{2+o(1)}(1/\eps)\right)\\
			&= n(\log^{2+o(1)}(\log n) + \log^{2+o(1)}(1/\eps))				
		\end{split}
	\end{equation}
\end{itemize}
	
\end{proof}

\begin{proof}[Proving the space bound in \Cref{thm:additiveOracleEasy}] Since $c = 24$,  by \Crefrange{lm:step1-space}{lm:step3-space}, the total space of the oracle in Steps 1-3 is :
	\begin{equation}\label{eq:space-1-3}
		O(\frac{n\eps^{c-1}}{\log n}  + \frac{n\eps^{c-13}}{\log n} + \frac{n\eps^{c-24}}{\log n})  = O(n)
	\end{equation}
	which is dominated by the total space of Step 4. Thus, the space bound of the oracle as claimed in \Cref{thm:additiveOracleEasy} follows from \Cref{lm:step4-space}.
\end{proof}

\subsubsection{Preprocessing}\label{subsubsec:preprocessing}

By \Cref{lm:T-size}, $\mathcal{T}$ can be constructed in $O(n\log n)$ time.  The most time-consuming step is to compute $R^+$ for each region $R$ associated with a node of $\mathcal{T}$ defined in Step 2. Recall that $R^+$ is obtained from $R$ by filling in the non-parental holes.  Note that $\mathcal{T}$ only has $O(n/\log(n)\eps^{c}) = O(n/\log(n))$ nodes. However, $R^+$ could have $\Omega(n)$ vertices, resulting in the total size of $\Omega(n^2/\log(n))$. A standard technique~\cite{Thorup04,KST13,WY16,CS19} to reduce the size of $R^+$ is to \emph{approximately} fill the holes of $R$ such that $|V(R^+)\setminus V(R)| = \poly(\log(n),1/\eps)$ and for every $u,v\in M(R)$, $d_R(u,v) \leq d_{R^+}(u,v)\leq d_R(u,v) + \eps D$.  We will show that the total size of all regions $\{R^+\}_{R\in \mathcal{T}}$ is $O(n \poly(\log(n)))$. (Indeed, we can reduce the total size of all regions to $O(n\log n\poly(1/\eps))$ using a more complicated adaptive filling technique of Chan and Skrepetos~\cite{CS19}.) Next, we describe the filling procedure in detail using \emph{dense portals}, following Weimann and Yuster~\cite{WY16}. 

\begin{figure}[!htb]
	\centering
	\includegraphics[width=.7\textwidth]{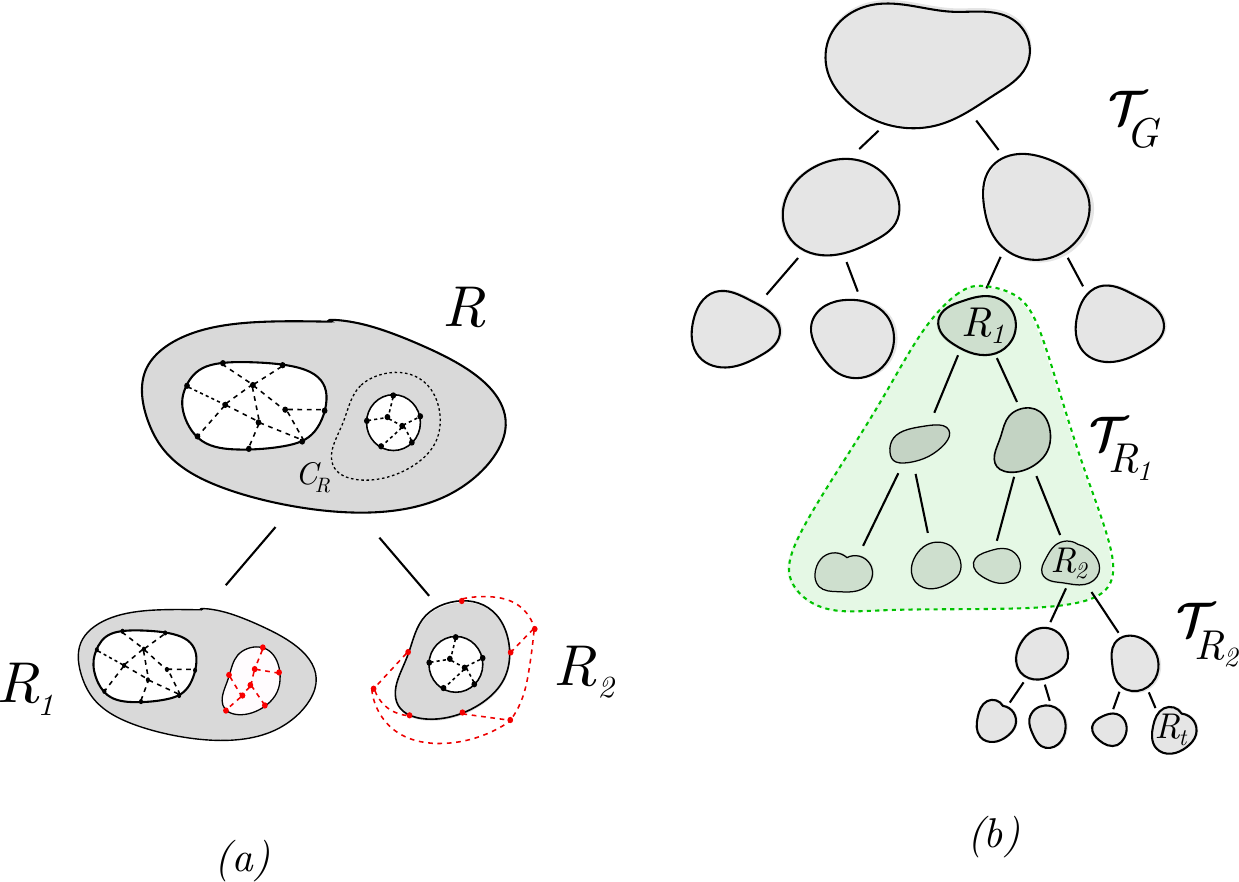}
	\caption{(a) A region $R$ with two holes; two child regions $R_1$ and $R_2$; distance preserving minors $H_1$ (associated with $R_1$) and $H_2$ (associated with $R_2$) compose of red vertices and edges. (b) An illustration for the proof of \Cref{thm:additiveOracle}; the highlighted subtree is the tree resulting from the recursive decomposition of a leaf region $R_1$ of the tree $\mathcal{T}_G$; the highlighted subtree corresponds to a node  at the second level in the hierarchy of oracles $\mathcal{H}$.}
	\label{fig:hierarchy}
\end{figure}

The filling process is top-down. The root of $\mathcal{T}$ is associated with $G$. Let $R$ be a region associated with a node of $\mathcal{T}$, and $R_1,R_2$ be its two child regions. Let $G_R$ be the approximately filled graph of $R$ where every hole was approximately filled. (Think of $G_R$ as an approximation of $G$ where the distances between marked vertices of $R$ are approximately preserved.) $G_R$ is given by induction; at the root, $G_R = G$.

Let $C_R$ be the shortest path separator that separates $R$ into $R_1,R_2$. We place $O(\eps/\log(n))$ portals, called \emph{dense portals}, on $C_R$ such that the distance between any two nearby dense portals is $c_1\eps D/\log(n)$ for some constant $c_1 \geq 1$. Next, we construct two distance preserving minors $H_1,H_2$ for the dense portals of $C_R$ as in Step 4(2) of the oracle construction, one for the subgraph of $G_R$ inside $C_R$ (and include $C_R$), and the other for  the subgraph of $G_R$ outside $C_R$ (and include $C_R$). See \Cref{fig:hierarchy}(a) for an illustration. Assume that $R_1$ is outside $C_R$ and $R_2$ is inside $C_R$. Then  $H_1$ and the subgraphs of $G_R$ outside $C_R$ form the approximately filled region of $R_1$, and similarly,  $H_2$ and the subgraphs of $G_R$ outside $C_R$ form the approximately filled region of $R_2$. Observe that $|V(H_1)| = O(\log^4(n)/\eps^4)$ and $|V(H_2)| = O(\log^4(n)/\eps^4)$ by \Cref{lm:dist-minor-size}. This means the total size of $G_{R_1}$ is $M(R_1) + O(\log^4(n)/\eps^4)$ (as $R_1$ has only 5 holes). The same holds for $R_2$. We called the vertices in $(V(G_{R_1})\cup V(G_{R_2}))\setminus (V(R_1)\cup V(R_2))$ \emph{Steiner vertices}. 

Since  $|V(\mathcal{T})| = O(n\eps^c/\log(n))$ and $c = 24$, the total number of Steiner vertices is $O(n\eps^c/\log(n))\cdot O(\log^4(n)/\eps^4) = O(n\log^3(n))$. Thus the total size of all approximately filled regions is $O(n\log^{3}n)$. For each graph $G_R$, computing $H_1$ and $H_2$ takes $O(\log^2(n)/\eps^2 |V(G_R)|)$ time, as it boils down to computing all-pairs shortest paths between dense portals of a hole that can be done $O(|V(G_R)|)$ time per shortest path~\cite{HKRS97}. (A more efficient way to compute  $H_1$ and $H_2$ in $O(n\log(n) + n\log^4(n)\eps^{c-4})$ is to use the multiple-source shortest path algorithm of Klein~\cite{Klein05B}; see Theorem 6 in~\cite{CS19}.)  The total time to compute all approximately filled regions is:
\begin{equation*}
	O((\sum_{R}|V(G_R)|)\log^2(n)/\eps^2 ) = O(n\log^{3}n \eps^{c} \log^2(n)/\eps^2) = O(n\log^{5}(n))~,
\end{equation*}
since $c = 24$.  

For each region $R$, and its approximately filled graph $G_R$, we can extract the approximate graph $R^+$ by removing all the Steiner vertices in the parental hole of $R$.

Next we compute the patterns of $R^+$ w.r.t. the sequence of portals $\sigma$ on the parental boundary of $R$. This can be done in $O(V(R^+)1/\eps)$ time by computing a shortest tree from each vertex in $\sigma$; there are only $O(1/\eps)$ vertices in the sequence $\sigma$. Given the patterns, all tables $T_{2}, T_{41a}, T_{41b}$ can be computed in $O(\poly(1/\eps))$ time per node. Thus, the total running time is $O(n\poly(1/\eps))$ since the number of nodes of $\mathcal{T}$ is $O(n\eps^c/\log(n)) = O(n)$.

Finally, by \Cref{thm:LongPettie}, the exact distance oracle for each leaf region $\mathcal{E}_R$ can be computed in time $|V(R')|^{3/2+o(1)} = \poly(\log n, 1/\eps)$ since $|V(R')| = O(\log(n)/\eps^{c})$. Thus, the total running time to compute all the exact distance oracles is $O(n\poly(\log n, 1/\eps))$

In summary, the total construction time of the oracle in \Cref{thm:additiveOracleEasy} is $O(n\poly(\log n,1/\eps))$.

We now show that the additive distortion due to approximate filling is $\eps D$.  Let $u,v$ be any two marked vertices in $R_1$ (a child region of $R$). Since the distance between any two nearby dense portals is at most  $c_1\eps D/\log(n)$, by the triangle inequality, $ d_{G_{R}}(u,v)\leq d_{G_{R_1}}(u,v)\leq d_{G_{R}}(u,v) + 2c_1\eps D/\log(n)$.  Since the depth of $\mathcal{T}$ is $O(\log n)$, by induction, it holds that  $d_{G_{R_1}}(u,v)\leq d_{G}(u,v) + O(\log n)2c_1\eps D/\log(n) = d_G(u,v) + \eps D$ for an appropriate choice of constant $c_1$. 
 
 Since the approximate filling incurs an additive distortion $\eps D$, the distance returned by the oracle is at most: 
\begin{equation*}
	(d_G(u,v) + \eps D) + \eps D  =  d_G(u,v) + 2\eps D
\end{equation*}
By scaling, we can get back additive distortion $+\eps D$.

\subsection{Reducing Space and Query Time: Proof of \Cref{thm:additiveOracle}}\label{subsec:additiveStrong}

We employ the bootstrapping idea of  Kawarabayashi, Sommer, and Thorup to replace $\log^{o(1)}(n)$ factor in the space and query time of the oracle in \Cref{thm:additiveOracleEasy} with a $\log^*(n)$ factor. With more careful analysis and using an LCA data structure to navigate the hierarchy  in the construction below, we save a $\log^*(n)$ factor in the query time and make the $\log^*(n)$ factor in the space \emph{additive} instead of multiplicative.  We restate \Cref{thm:additiveOracle} below for convenience.

\AdditiveOracle*
\begin{proof}  Starting from $G$, we apply all steps from 1-4 of the construction in \Cref{subsec:easyOracle} to $G$, creating a recursive decomposition tree $\mathcal{T}_G$, except that in Step 4(2), we do not construct an exact distance oracle $\mathcal{E}_{R_1}$ for each leaf region $R_1$. Instead, we recurse on $R_1$, creating a second-level recursive decomposition tree $\mathcal{T}_{R_1}$. That is, we apply all steps from  1-4 to $R_1$, except Step 4(2). See \Cref{fig:hierarchy}(b) for an illustration of the construction. Note that $|M(R_1)| = \Theta(\log(n)/\eps^c)$ for $c=24$. Let $\mathcal{T}_{R_1}$ be the recursion tree induced by the recursive decomposition of $R_1$; leaves of $\mathcal{T}_{R_1}$ correspond to regions, say $R_2$, of $R_1$ that have  $|M(R_2)| = \Theta(\log\log(n)/\eps^{c})$.  We then continue to recurse on $R_2$. Generally, at step $j$ of the recursion, the number of marked vertices in the region associated with each leaf $R_j$ of the recursive decomposition tree, denoted by $\mathcal{T}_{R_{j-1}}$, is $O(\frac{\log^{(j)}n}{\eps^c})$ where $\log^{(j)}(n) = \log(\log(\cdots \log(n)\cdots))$; the logarithm is applied $j$ times.  We stop the recursion when a leaf region, say $R_t$, at some step $t$ of the recursion, has $|M(R_t)| = \Theta(1/\eps^{c})$.	We apply \Cref{thm:LongPettie} to construct an exact distance oracle $\mathcal{E}_{R_t}$ for the contracted filled graph $R'_t$ of $R_t$ with:
	\begin{itemize}
		\item \textbf{Regime 4(2a).~}  $|V(R'_t)|^{1+o(1)} = |V(R'_t)|(1/\eps)^{o(1)}$ space  and  $O(\log^{2+o(1)}(|V(R'_t)|)) = \log^{2+o(1)}(1/\eps)$ query time or
		\item \textbf{Regime 4(2b).~} $|V(R'_t)|\log^{2+o(1)}(|V(R'_t)|) = |V(R'_t)|\log^{2+o(1)}(1/\eps)$ space and $|V(R'_t)|^{o(1)} = (1/\eps)^{o(1)}$ query time. 
	\end{itemize}
	Note that $|V(R'_t)| = \poly(1/\eps)$ since $|M(R_t)| = \Theta(1/\eps^c)$.
	
	This construction gives us a \emph{hierarchy} $\mathcal{H}$ of oracles: each recursion step corresponds to a level in $\mathcal{H}$. Each internal node $\tau$ of the hierarchy at level $j$ corresponds to the oracle, say $\mathcal{O}_{R_{j-1}}$, for a leaf region $R_{j-1}$ at level $j-1$ (when $j = 1$, we denote $R_0 = G$). The oracle  $\mathcal{O}_{R_{j-1}}$ allows us to query distances between marked vertices that are not in the same leaf of the recursive decomposition tree $\mathcal{T}_{R_{j-1}}$ in $O(1)$ time, following the analysis in \Cref{subsubsec:time-stretch}. Each leaf of $\mathcal{H}$ corresponds to an exact distance oracle for a region $R_t$.  
	
	Clearly, the recursion depth, which is also the depth of the hierarchy $\mathcal{H}$, is $t = O(\log^{*}n)$ since each time we recurse, the size of the region is reduced from $O(\frac{\log^{(j-1)}(n)}{\eps^c})$ to $O(\frac{\log^{(j)}n}{\eps^c})$ for some $j \in \{2,3,\ldots, t\}$; in the first level ($j = 1$), the size is reduced from $n$ to $O(\log(n)/\eps^{c})$. 
	
	By the analysis in \Cref{subsubsec:space-analysis}, in particular \Cref{eq:space-1-3}, the total space of each \emph{non-leaf level} of $\mathcal{H}$ is $O(n)$. Thus, the total space of $\mathcal{H}$ associated with non-leaf nodes  is $O(n\log^*n)$.  On the other hand, by the same analysis in \Cref{eq:space-R1,eq:space-R2}, the total space of the oracles at leaves of $\mathcal{H}$ is:
		\begin{itemize}
		\item  $\sum_{R_t \text{ is a leaf of }\mathcal{H}}|V(R'_t)|(1/\eps)^{o(1)} = n\eps^{o(1)}$ space in Regime 4(2a)  or
		\item $\sum_{R_t \text{ is a leaf of }\mathcal{H}}  |V(R'_t)|\log^{2+o(1)}(1/\eps) = n \log^{2+o(1)}(1/\eps) $ in Regime 4(2b).
	\end{itemize}
	Thus, the total space of the oracle is $O(n(\eps^{o(1)} + \log^*n))$ in Regime 4(2a) and is  $O(n(\log^{2+o(1)}(1/\eps) + \log^*n))$ in Regime 4(2b), as claimed.
	
	To answer a query quickly, we augment $\mathcal{H}$ with the following: for each vertex $v\in G$, we store a pointer to a leaf node of $\mathcal{H}$ whose corresponding region $R_t$ contains $v$ as a marked vertex. Furthermore, we construct an LCA data structure for $\mathcal{H}$. Note that $\mathcal{H}$ has  $O(n)$ nodes as it is a binary tree with at most $n$ leaves, the total space augmented to $\mathcal{H}$ is $O(n)$.
	
	Now given two vertices $u$ and $v$, let $R_t(u)$ and $R_t(v)$ be two leaf regions of $\mathcal{H}$ containing $u$ and $v$.  If $R_t(u) \not= R_t(v)$, we query the lowest common ancestor, denoted by $R_{uv}$, of $R_t(u)$ and $R_t(v)$ in $O(1)$ time. Then, the approximate distance query can be done in $O(1)$ by querying $\mathcal{O}_{R_{uv}}$.  If $R_t(u) = R_t(v) = R$, we query the exact distance oracle $\mathcal{E}_{R}$ to obtain an approximate distance between $u$ and $v$ in time $\log^{2+o(1)}(1/\eps)$ in Regime 4(2a) and in time $(1/\eps)^{o(1)}$ in Regime 4(2b). Following the stretch analysis in \Cref{subsubsec:time-stretch}, the additive stretch is $+O(\eps)D$. 
	
	For the construction time, by \Cref{thm:additiveOracleEasy}, each level of $\mathcal{H}$ can be constructed in time $n\poly(\log n,\eps)$ time. Since the depth of $\mathcal{H}$ is $O(\log^*n)$, the running time to construct $\mathcal{H}$ is $n\poly(\log n,\eps) \cdot \log^* n = n\poly(\log n,\eps)$, as desired.
\end{proof}

\section{Distance Oracles with Multiplicative Stretch: Proof of \Cref{thm:main}}\label{sec:mulitiplicative}

The construction relies on \emph{sparse covers} as defined below. For a graph $G$, we denote by $\diam(G)$ the diameter of $G$. For a vertex $v\in V(G)$ and a parameter $r > 0$, we denote by $B_G(v,r) = \{u\in V(G): d_G(u,v)\leq r\}$ the ball of radius $r$ centered at $v$.

\begin{definition}[Sparse Cover] \label{def:sparse-cover}   A $(\beta,s,\Delta)$-sparse cover of an edge-weighted graph $G = (V,E,w)$ is a collection of  \emph{induced subgraphs} $\mathcal{C} = \{C_1,\ldots, C_k\}$, called \emph{clusters} such that:
	\begin{itemize}[noitemsep]
		\item[(1)] $\diam(C_i) \leq \Delta$ for every $i\in [k]$.
		\item[(2)] For every $v\in V$, there exists $i\in [k]$ such that $B_G(v,\Delta/\beta)\subseteq V(C_i)$.
		\item[(3)] Every vertex $v$ is contained in at most $s$ clusters in $\mathcal{C}$.
	\end{itemize}
 \end{definition}
If for any given $\Delta > 0$, $G$ has a  $(\beta, s,\Delta)$-sparse cover, we say that $G$ admits a  \emph{$(\beta,s)$-sparse covering scheme}.

The notion of sparse covers was introduced by Awerbuch and Peleg~\cite{AP90}.  Busch, LaFortune, and Tirthapura~\cite{BLT07}  showed that planar graphs admit an  $(O(1),O(1))$-sparse covering scheme.  Abraham, Gavoille, Malkhi, and Wieder~\cite{AGMW10} extended the result of Busch, LaFortune, and Tirthapura~\cite{BLT07} to minor-free graphs. Le and Wulff-Nilsen~\cite{LW21} showed that a sparse cover of planar graphs can be constructed in linear time.

\begin{lemma}[Lemma 1 in the full version of~\cite{LW21}]\label{lm:sparse-cover} Given a planar graph $G= (V,E,w)$ with $n$ vertices and any parameter $\Delta > 0$, then one can construct an  $(O(1),O(1),\Delta)$-sparse cover of $G$ in $O(n)$ time.
\end{lemma}

\begin{lemma}\label{lm:oracle-interval} Let $\eps \in (0,1), r > 0$ be positive parameters and $G = (V,E,w)$ be an $n$-vertex planar graph. We can construct in $O(n\poly(\log(n),1/\eps))$ time an oracle $\mathcal{O}_G$ such that: 
	\begin{equation}
		\begin{aligned}
			&d_G(u,v)\leq \mathcal{O}_G(u,v)  \quad && \mbox{for all }u,v\in V\\	
			&\mathcal{O}_G(u,v)\leq (1+\eps)d_G(u,v)  &&\mbox{for $u,v\in V$ s.t. }d_G(u,v)\in [r, 2r]			
		\end{aligned}
	\end{equation} 
Here $\mathcal{O}_G(u,v)$ is the distance returned by $\mathcal{O}_G$. Furthermore, 	$\mathcal{O}_G$  has
	\begin{itemize}
		\item[(1)]  $O(n((1/\eps)^{o(1)} + \log^*(n)))$ space  and  $\log^{2+o(1)}(1/\eps)$ query time or 
		\item[(2)] $O(n(\log^{2+o(1)}(1/\eps) + \log^*(n)))$ space  and  $(1/\eps)^{o(1)}$ query time.
	\end{itemize}
\end{lemma}
\begin{proof}
	We construct a $(\beta,s,\beta\cdot (2r))$-sparse cover  $\mathcal{C}$ for $G$ with $\beta = O(1)$ and $s = O(1)$ using \Cref{lm:sparse-cover}.  Since $s = O(1)$, by property (3) of \Cref{def:sparse-cover}, we have:
	\begin{equation}\label{eq:cover-size}
	\sum_{C\in \mathcal{C}} |V(C)| \leq s |V| = O(n)
	\end{equation}
	
	For each cluster $C\in \mathcal{C}$, we apply \Cref{thm:additiveOracle} to construct an additive distance oracle $\mathcal{O}_C$ with additive stretch $+ \eps \diam(C) $. The oracle for $G$ consists of the oracles for all clusters in $\mathcal{C}$.	 In addition, for each vertex $u$, we will store $O(1)$ pointers to each cluster $C\in \mathcal{C}$ that contains $u$. If we let $\mathcal{C}_{uv}\subseteq \mathcal{C}$ be the set of $O(1)$ clusters containing both $u$ and $v$, then the approximate distance between $u$ and $v$ is:
	\begin{equation}\label{eq:dist-cover}
	\mathcal{O}_G(u,v) = \min_{C\in \mathcal{C}_{uv}}\mathcal{O}_C(u,v)
	\end{equation}
	where $\mathcal{O}_C(u,v)$ is the distance returned by the oracle $\mathcal{O}_C$. (If $\mathcal{C}_{u,v} = \emptyset$, we simply set $\mathcal{O}_G(u,v) = +\infty$.) 
	
	Clearly, $\mathcal{O}_G(u,v)\geq d_G(u,v)$ for any $u,v\in V$ since $d_C(u,v)\leq d_G(u,v)$ for any subgraph $C$ of $G$. 
	
	Next,  we consider the case where $d_G(u,v)\in [r,2r]$.  By property (2) in \Cref{def:sparse-cover}, there is a cluster $X \in \mathcal{C}$ such that $B_G(v, (\beta 2r)/\beta) = B_G(v,2r) \subseteq V(X)$. Since $X$ is an induced subgraph of $G$, it holds that $d_X(u,v) = d_G(u,v)$. Furthermore, since the additive stretch of $\mathcal{O}_X$ is :
	\begin{equation*}
	\eps \diam(X) \leq \eps (2\beta r) = O(\eps)r   = O(\eps)d_G(u,v)~,
	\end{equation*}
	it follows that:
	\begin{equation*}
	\mathcal{O}_X(u,v) \leq  d_X(u,v) + O(\eps)d_G(u,v) = (1+O(\eps))d_G(u,v)~.
	\end{equation*}
	This and \Cref{eq:dist-cover} imply that $\mathcal{O}(u,v)  \leq  (1 + O(\eps))d_G(u,v)$. By scaling $\eps$, we get that $\mathcal{O}(u,v)  \leq  (1 + \eps)d_G(u,v)$, as claimed.
	
	We now analyze the space and query time of $\mathcal{O}$. 
	
	If  we use Regime (1) in \Cref{thm:additiveOracle} to construct  $\mathcal{O}_C$, then the query time is $\log^{2+o(1)}(1/\eps)$ and the total space is:
	\begin{equation*}
	\sum_{C\in \mathcal{C}} O(|V(C)|)((1/\eps)^{o(1)} + \log^*(|V(C)|)) = O(n((1/\eps)^{o(1)} + \log^*(n)))
	\end{equation*}
	by \Cref{eq:cover-size}.

	If we use Regime (2) in \Cref{thm:additiveOracle} to construct $\mathcal{O}_C$, then the query time is $(1/\eps)^{o(1)}$ and the total space is:
	\begin{equation*}
	\sum_{C\in \mathcal{C}} O(|V(C)|)(\log^{2+o(1)}(1/\eps)+ \log^*(|V(C)|)) = O(n(\log^{2+o(1)}(1/\eps) + \log^*(n)))
	\end{equation*}
	by \Cref{eq:cover-size}.
	
	For the construction time, by \Cref{lm:sparse-cover}, $\mathcal{C}$ can be constructed in $O(n)$ time. By \Cref{thm:additiveOracle}, $\mathcal{O}_C$ can be constructed in $|V(C)|\poly(\log(|V(C)|),\eps ) = |V(C)|\poly(\log(n),\eps)$ time. Thus, by \Cref{eq:cover-size}, the total running time to construct all oracles $\mathcal{O}_C$ is $O(n\poly(\log(n),\eps))$ as claimed.
\end{proof}

We are now ready to prove \Cref{thm:main} that we restate below for convenience.

\Main*
\begin{proof} By scaling edge weights, we assume that the minimum distance is $1$. For each $i = 0,1,2,\ldots$, we denote $r_i = 2^i$. Let $G_i$ be obtained from $G$ by contracting every edge of weight at most $(r_i\eps)/n$. Observe that, for every pair $(u,v)$ such that $d_G(u,v) \in [r_i, 2r_i]$, we have:
	\begin{equation}\label{eq:dG-vs-dGi}
		d_G(u,v) - \eps r_i\leq d_{G_i}(u,v)\leq d_G(u,v).
	\end{equation}
	This is because we contract at most $n$ edges of weight at most $r_i/(\eps n)$ each and, hence, the distance loss due to the contraction is at most $n \cdot (r_i\eps)/n \leq \eps r_i $. Furthermore, by construction, each edge $e\in G$ belongs to at most $O(\log n)$ graphs $G_i$; it follows that:
	\begin{equation}\label{eq:Gi-size}
		\sum_{i\geq 0} |E(G_i)| = O(n\log n)
	\end{equation} 
	For each subgraph $G_i$, we apply \Cref{lm:oracle-interval} to construct a distance oracle $\mathcal{O}_{G_i}$. 
	
 	Next, we construct a $2$-approximate distance oracle $\mathcal{O}_2$ with $O(n\log n)$ space and $O(1)$ query time; such an oracle can be constructed in $O(n\log^{3}(n))$ time by applying the construction of Thorup~\cite{Thorup04} and Klein~\cite{Klein02} with $\eps = 1$.
	
	Our final oracle, denoted by $\mathcal{O}$, consists of $\mathcal{O}_2$ and all oracles $\{\mathcal{O}_{G_i}\}_{i\geq 0}$. 
	
	To query $\mathcal{O}$ given two vertices $u$ and $v$, first we query $\mathcal{O}_2$ to get a 2-approximation of $d_G(u,v)$, denoted by $\mathcal{O_2}(u,v)$. Then we compute a set of 3  indices $I_{uv} = \{i_0-2,i_0-1, i_0\}$ with $i_0 = \lfloor  \log_2(\mathcal{O}_2(u,v))\rfloor$. Finally, for each index $j \in I_{uv}$, we query the oracle $\mathcal{O}_{G_j}$ and return:
	\begin{equation}\label{eq:returned-dist}
		\mathcal{O}(u,v) = \min_{j \in I_{uv}}\{\mathcal{O}_{G_j}(u,v) + \eps r_j\}
	\end{equation}
	
	We now bound the stretch of $\mathcal{O}$.  Let $i_{uv}$ be such that $d_G(u,v)\in [r_{i_{uv}}, 2r_{i_{uv}})$.  This means that if we query the oracle $\mathcal{O}_{G_{i_{uv}}}$, by \Cref{lm:oracle-interval}, the returned distance $\mathcal{O}_{G_{i_{uv}}}(u,v) $ satisfies:
	\begin{equation*}
		d_{G_{i_{uv}}}(u,v)\leq \mathcal{O}_{G_{i_{uv}}}(u,v) \leq (1+\eps)d_{G_{i_{uv}}}(u,v)
	\end{equation*}
	Thus, from \Cref{eq:dG-vs-dGi} and the fact that $d_G(u,v)\geq r_{i_{uv}}$, we have:
	\begin{equation*}
		\begin{split}
			&d_G(u,v)\leq d_{G_{i_{uv}}}(u,v) +  \eps r_{uv} \leq \mathcal{O}_{G_{i_{uv}}}(u,v) + \eps r_{uv} \\
			&\mathcal{O}_{G_{i_{uv}}}(u,v) + \eps r_{uv} \leq (1+\eps)d_{G_{i_{uv}}}(u,v) + \eps r_{uv} \leq (1+2\eps)d_{G}(u,v),
		\end{split}
	\end{equation*}
	implying that the (multiplicative) stretch of $\mathcal{O}$ is $(1+2\eps)$; by scaling $\eps$, we get back stretch $(1+\eps)$.
	
	We now bound the space and query time of $\mathcal{O}$; we consider two regimes in \Cref{lm:oracle-interval} that we use to construct $\mathcal{O}_i$. 
	\begin{enumerate}
		\item \textbf{Regime (1)~.} Since the query time of $\mathcal{O}_2$ is $O(1)$, $I_{uv}$ can be computed in $O(1)$ time. Since $|I_{uv}|= 3$ and the query time of each $\mathcal{O}_{G_i}$ is $\log^{2+o(1)}(1/\eps)$, the total query time is $\log^{2+o(1)}(1/\eps)$. The space of $\mathcal{O}_2 = O(n\log n)$ and the total space of all $\{\mathcal{O}_{G_i}\}_{i\geq 0}$, by \Cref{lm:oracle-interval}, is:
		\begin{equation*}
			 \sum_{i\geq 0} O(|V(G_i)|((1/\eps)^{o(1)} + \log^*(|V(G_i)|))) = O(n\log(n) ((1/\eps)^{o(1)} + \log^*n))
		\end{equation*}
		by \Cref{eq:Gi-size}. This implies the claimed space bound.

		\item \textbf{Regime (2)~.} In this regime, the query time of each $\mathcal{O}_{G_i}$ is $(1/\eps)^{o(1)}$ which is also the total query time. The total space of all $\{\mathcal{O}_{G_i}\}_{i\geq 0}$, by \Cref{lm:oracle-interval} is:
		\begin{equation*}
			\sum_{i\geq 0} O(|V(G_i)|(\log^{2+o(1)}(1/\eps) + \log^*(|V(G_i)|))) = O(n\log(n) (\log^{2+o(1)}(1/\eps)  + \log^*n))
		\end{equation*}
		by \Cref{eq:Gi-size}, as desired.
	\end{enumerate}

	For the construction time, recall that the construction time of $\mathcal{O}_2$  is $O(n\log^3(n))$. The construction time of each $\mathcal{O}_{G_i}$ is $|V(G_i)|\poly(\log(|V(G_i)|),1/\eps)  = |V(G_i)| \poly(\log(n),1/\eps)$. By \Cref{eq:Gi-size}, the total construction time of $\mathcal{O}$ is $n\poly(\log(n),1/\eps)$.
\end{proof}

\paragraph{Acknowledgement.~} This work is supported by the National Science Foundation under Grant No. CCF-2121952. We thank Christian Wulff-Nilsen for many helpful conversations.
	\bibliographystyle{alphaurlinit}
	\bibliography{spanner}
	
	\pagebreak
	\appendix

\end{document}